\documentclass[nohyperref]{article}

\usepackage{microtype}
\usepackage{graphicx}
\usepackage{subfigure}
\usepackage{booktabs} 

\usepackage{hyperref}

\usepackage[accepted]{icml2022}

\usepackage[mathscr]{eucal} 
\usepackage{amsmath}
\usepackage{amssymb}
\usepackage{mathtools}

\usepackage{bm}
\usepackage{thmtools}
\usepackage{thm-restate}
\usepackage{nameref}
\usepackage{hyperref}
\usepackage{amsthm}
\usepackage[capitalize, nameinlink,noabbrev]{cleveref}
\usepackage{mathtools}
\usepackage{nag}
\usepackage{authblk}
\usepackage{stmaryrd}
\usepackage{bbm}
\usepackage[skins]{tcolorbox}
\usepackage{xfrac}

\usepackage{tikz}
\usetikzlibrary{positioning}

\newcommand{\declarecolor}[2]{\definecolor{#1}{RGB}{#2}\expandafter\newcommand\csname #1\endcsname[1]{\textcolor{#1}{##1}}}
\declarecolor{White}{255, 255, 255}
\declarecolor{Black}{0, 0, 0}
\declarecolor{LightGray}{216, 216, 216}
\declarecolor{Gray}{127, 127, 127}
\declarecolor{Orange}{237, 125, 49}
\declarecolor{LightOrange}{251,229, 214}
\declarecolor{Yellow}{255, 192, 0}
\declarecolor{LightYellow}{255, 242, 200}
\declarecolor{Blue}{91, 155, 213}
\declarecolor{LightBlue}{222, 235, 247}
\declarecolor{Green}{112, 173, 71}
\declarecolor{LightGreen}{226, 240, 217}
\declarecolor{Navy}{68, 114, 196}
\declarecolor{LightNavy}{218, 227, 243}

\crefformat{equation}{(#2#1#3)}

\newcommand{\Loss}{L}
\newcommand{\SurrogateLoss}{\widetilde{\Loss}}
\newcommand{\MaxObj}{f}

\newcommand{\mat}[1]{\mathbf{#1}}

\newcommand{\tensor}[1]{\bm{\mathscr{#1}}}

\newcommand{\frobenius}{\textnormal{F}}

\newcommand{\ktimes}{\otimes}

\newcommand{\rank}{\textnormal{rank}}

\newcommand{\R}{\mathbb{R}}
\newcommand{\Z}{\mathbb{Z}}

\DeclareMathOperator*{\defeq}{\overset{def}{=}}
\renewcommand{\epsilon}{\varepsilon}

\DeclareMathOperator*{\argmax}{arg\,max}
\DeclareMathOperator*{\argmin}{arg\,min}

\newcommand{\ropti}[1]{R_{#1}^*}

\DeclarePairedDelimiter{\inner}{\langle}{\rangle}
\DeclarePairedDelimiter{\abs}{\lvert}{\rvert}
\DeclarePairedDelimiter{\set}{\{}{\}}
\DeclarePairedDelimiter{\parens}{(}{)}
\DeclarePairedDelimiter{\bracks}{[}{]}
\DeclarePairedDelimiter{\norm}{\lVert}{\rVert}
\DeclarePairedDelimiter{\ceil}{\lceil}{\rceil}
\DeclarePairedDelimiter{\floor}{\lfloor}{\rfloor}

\usepackage{amsmath}
\usepackage{amssymb}
\usepackage{mathtools}

\theoremstyle{plain}
\newtheorem{theorem}{Theorem}[section]

\newtheorem{lemma}[theorem]{Lemma}

\theoremstyle{definition}
\newtheorem{definition}[theorem]{Definition}

\theoremstyle{remark}
\newtheorem{remark}[theorem]{Remark}

\newcommand{\todo}[1]{{\color{red} TODO: {#1}}}

\icmltitlerunning{Approximately Optimal Core Shapes for Tensor Decompositions}

\begin{document}

\twocolumn[
\icmltitle{Approximately Optimal Core Shapes for Tensor Decompositions}



\icmlsetsymbol{equal}{*}

\begin{icmlauthorlist}
\icmlauthor{Mehrdad Ghadiri}{equal,gt}
\icmlauthor{Matthew Fahrbach}{equal,google}
\icmlauthor{Gang Fu}{google}
\icmlauthor{Vahab Mirrokni}{google}
\end{icmlauthorlist}

\icmlaffiliation{google}{Google Research}
\icmlaffiliation{gt}{Georgia Tech}

\icmlcorrespondingauthor{Mehrdad Ghadiri}{mghadiri3@gatech.edu}

\icmlkeywords{Machine Learning, ICML}

\vskip 0.3in
]



\printAffiliationsAndNotice{\icmlEqualContribution} 

\begin{abstract}
This work studies the combinatorial optimization problem of
finding an optimal \emph{core tensor shape}, also called multilinear rank,
for a size-constrained Tucker decomposition.
We give an algorithm with provable approximation guarantees
for its reconstruction error via connections to higher-order singular values.
Specifically, we introduce a novel \emph{Tucker packing problem},
which we prove is NP-hard,
and give a polynomial-time approximation scheme
based on a reduction to the 2-dimensional knapsack problem with a matroid constraint.
We also generalize our techniques to \emph{tree tensor network decompositions}.
We implement our algorithm using an integer programming solver,
and show that its solution quality is
competitive with (and sometimes better than)
the greedy algorithm that uses the true Tucker decomposition loss at each step,
while also running up to 1000x faster.
\end{abstract}

 \section{Introduction}
\label{sec:introduction}

Low-rank tensor decomposition is a powerful tool in
the modern machine learning toolbox.
Like low-rank matrix factorization,
it has countless applications in
scientific computing, data mining, and signal processing~\citep{kolda2009tensor,sidiropoulos2017tensor},
e.g., anomaly detection in data streams~\cite{jang2021fast}
and
compressing convolutional
neural networks on mobile devices for
faster inference while reducing power consumption~\citep{kim2015compression}.

The most widely used tensor decompositions are
the canonical polyadic (CP) decomposition,
Tucker decomposition,
and tensor-train decomposition~\citep{oseledets2011tensor}---the last two being instances of
\emph{tree tensor networks}~\citep{kramer2020tree}.
CP decomposition factors a tensor into the sum of~$r$ rank-one tensors.
Tucker decomposition, however,
specifies the rank $R_n$ in each dimension $n$
and relies on a core tensor $\tensor{G} \in \R^{R_1 \times \dots \times R_N}$
for reconstructing the decomposition.
The notion of \emph{multilinear rank} $\mat{r} = (R_1,\dots,R_N)$
puts practitioners in a challenging spot
because the set of feasible core shapes can be exponentially large.
Furthermore, searching in this state space can be prohibitively expensive
because evaluating the true quality of a core shape requires computing a Tucker decomposition,
which for large tensors can take hours and consume hundreds of GB of RAM.
For example, in the MATLAB Tensor Toolbox~\citep{matlab},
we need to specify the core shape parameter
\texttt{ranks} in advance before computing a size-constrained Tucker decomposition.

\begin{figure}[t]
\centering
\includegraphics[width=0.75\linewidth]{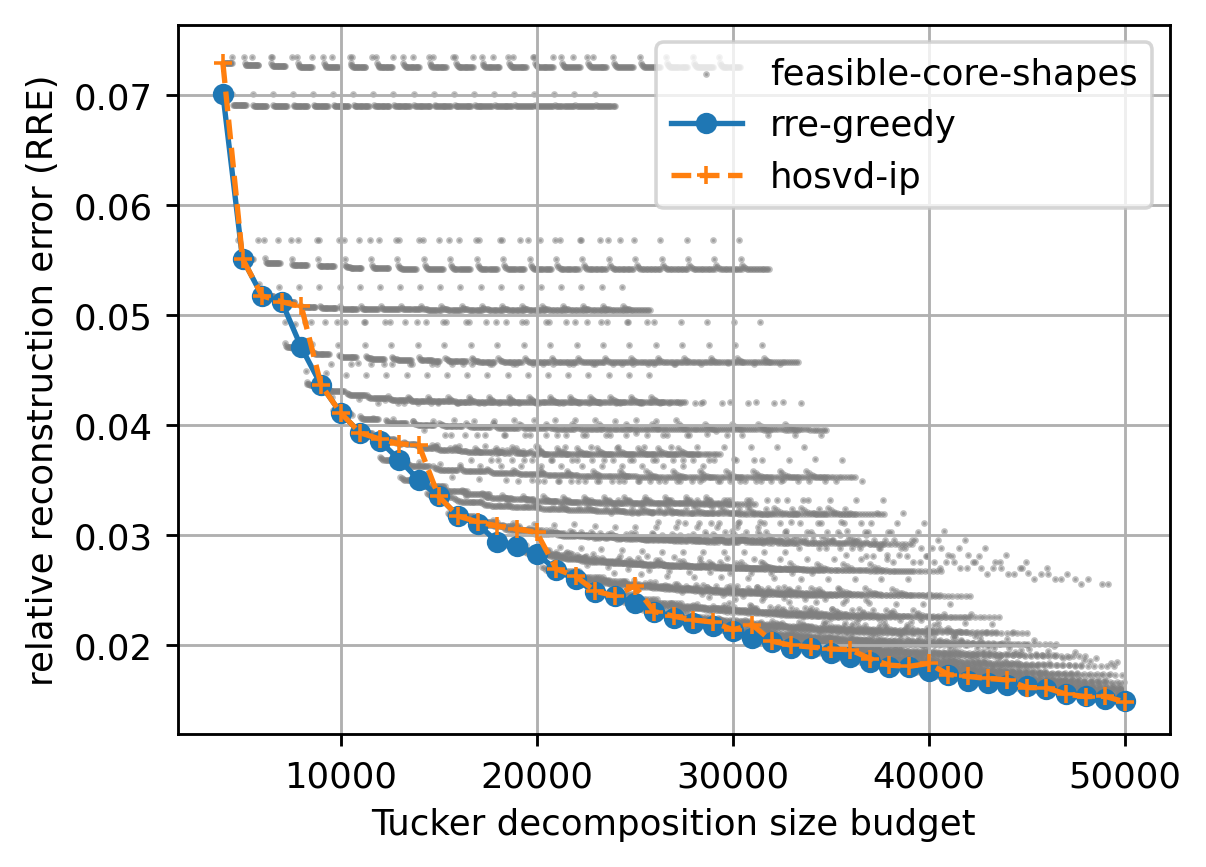}
\vspace{-0.25cm}
\caption{Pareto frontier of core shapes $\mat{r} \in [20]^3$
for hyperspectral tensor $\tensor{X} \in \R^{1024 \times 1344 \times 33}$.
Plots the RRE, i.e., $L(\tensor{X},\mat{r}) / \norm{\tensor{X}}_{\frobenius}^2$, as a function of compression rate.
RRE-greedy builds core shapes
by computing Tucker decompositions at each step.
HOSVD-IP is~\Cref{alg:total-size-Tucker} with integer programming,
which builds core shapes
via a surrogate packing problem on higher-order singular values.
}
\label{fig:pareto-curve}
\vspace{-0.25cm}
\end{figure}

In practice, the most popular Tucker decomposition algorithms are the 
$\mat{r}$-truncated higher-order singular value decomposition (HOSVD) in~\citet{HOSVD},
sequentially truncated ST-HOSVD in~\citet{vannieuwenhoven2012new},
and higher-order orthogonal iteration (HOOI),
which is a structured alternating least squares algorithm.

We explore the simple but fundamental discrete optimization problem for low-rank tensor decompositions:
\begin{quote}\emph{If a Tucker decomposition of $\tensor{X}$ can use at most~$c$ parameters,
which core tensor shape minimizes the reconstruction error?}
\end{quote}
This is a multilinear generalization of the best rank-$r$ matrix approximation problem.
While there are many parallels to low-rank matrix factorization,
tensor rank-related problems can be thoroughly different and more challenging than their
matrix counterparts.
For example, computing the CP rank of a real-valued tensor is NP-hard~~\citep{hillar2013most}.

\subsection{Our contributions and techniques}
We summarize the main contributions of this work below:
\begin{enumerate}
    \item We formalize the core tensor shape problem for
    size-constrained Tucker decompositions and introduce the
    \emph{Tucker packing problem}, which we prove is NP-hard.
    The approximation algorithms we
    develop build on a relationship between the optimal
    reconstruction error of a rank-$\mat{r}$
    Tucker decomposition and a multi-dimensional tail sum of its
    higher-order singular values~\citep{HOSVD, hackbusch2019tensor}.

    \item We design a polynomial-time approximation scheme (PTAS)
    for the surrogate Tucker packing problem (\Cref{thm:budget-splitting})
    by showing that it suffices to consider a small number of
    budget splits between the cost of
    the core tensor and the cost of the factor matrices.
    Each budget split subproblem reduces to a
    \emph{2-dimensional knapsack problem with a partition matroid constraint}
    after minor transformations.
    We solve these knapsack problems
    using the PTAS of~\citet{grandoni2014new},
    or in practice with integer linear programming.
    
    \item We extend our approach to
    \emph{tree tensor networks},
    which generalize the Tucker decomposition,
    tensor-train decomposition, and hierarchical
    Tucker decomposition.
    In doing so, we synthesize several works
    on tree tensors from the mathematics and physics communities,
    and give a succinct introduction for computer scientists.
    
    \item Finally,
    we demonstrate the effectiveness of our Tucker
    packing-based core shape solvers on four real-world tensors.
    Our HOSVD-IP algorithm
    is competitive with (and sometimes outperforms)
    the greedy algorithm that uses the true RRE,
    while running up to 1000x faster.
\end{enumerate}

\subsection{Related works}
\paragraph{Core shape constraints.}
\citet{de2000best} introduced the problem
of computing the best rank-$\mat{r}$
tensor approximation for a \emph{prespecified} core shape~$\mat{r}$,
and demonstrated the benefit of
initializing the decomposition with a
truncated HOSVD
and then running iterative methods such as HOOI.
\citet{elden2009newton, ishteva2009differential, ishteva2011best, ishteva2013jacobi, elden2022krylov}
consider this problem for rank-$(r_1,r_2,r_3)$
decompositions and develop a suite of advanced algorithms:
a Newton method on Grassmannian manifolds,
a trust-region method on Riemannian manifolds,
Jacobi rotations for symmetric tensors,
and a Krylov-type iterative method.
All these works, however, are concerned with optimizing the
tensor decomposition for a fixed core shape---not with optimizing the core tensor shape itself.


\citet{ehrlacher2021adaptive} and \citet{xiao2021rank}
recently explored \emph{rank-adaptive} methods for HOOI that
find minimal core shapes such that the Tucker decomposition achieves
a target reconstruct error.
They also leverage properties of the HOSVD, but they
\emph{do not impose a hard constraint on the size of the returned Tucker decomposition}.
\citet{hashemizadeh2020adaptive} generalized the
RRE-greedy algorithm in~\Cref{fig:pareto-curve}
to tensor networks for both rank and size constraints.

\paragraph{Low-rank tensor decomposition.}
\citet{song2019relative} gave
polynomial-time
$(1+\varepsilon)$-approximation algorithms
for many types of low-rank tensor decompositions
with respect to the Frobenius norm, including CP and Tucker decompositions.
\citet{frandsen2022optimization} showed that if a third-order tensor
has an exact Tucker decomposition, then
all local minima of an appropriately regularized
loss landscape are globally optimal.
Several works recently studied
Tucker decomposition in streaming models~\citep{traore2019singleshot, sun2020low}
and a sliding window model~\citep{jang2021fast}.
Fast randomized low-rank tensor decomposition
algorithms based on sketching have been proposed in~\citet{zhou2014decomposition,cheng2016spals,battaglino2018practical,malik2018low, che2019randomized, ma2021fast, larsen2022practical, fahrbach2022subquadratic, malik2022more}.

\section{Preliminaries}
\label{sec:preliminaries}


\paragraph{Notation.}
The \emph{order} of a tensor is its number of dimensions.
We denote scalars by normal lowercase letters $x \in \R$,
vectors by boldface lower letters $\mat{x} \in \R^n$,
matrices by boldface uppercase letters $\mat{X} \in \R^{m \times n}$,
and higher-order tensors by boldface script letters
$\tensor{X} \in \R^{I_1 \times \dots \times I_N}$.
We use normal uppercase letters for the size of an index set,
e.g., $[N] = \{1, 2, \dots, N\}$.
We denote the $i$-th entry of vector~$\mat{x}$ by $x_i$,
the $(i,j)$-th entry of matrix $\mat{X}$ by $x_{ij}$, and the
$(i,j,k)$-th entry of a third-order tensor $\tensor{X}$ by $x_{ijk}$.

\paragraph{Tensor products.}
The \emph{fibers} of a tensor are the vectors we get by
fixing all but one index.
For example, if $\tensor{X} \in \R^{3}$,
we denote the column, row, and tube fibers by
$\mat{x}_{:jk}$, $\mat{x}_{i:k}$, and $\mat{x}_{ij:}$, respectively.
The \emph{mode-$n$ unfolding} of a
tensor $\tensor{X} \in \R^{I_1 \times \dots \times I_N}$ is the matrix
$\mat{X}_{(n)} \in \R^{I_n \times (I_1 \dots I_{n-1} I_{n+1}\dots I_N)}$
that arranges the mode-$n$ fibers of $\tensor{X}$ as columns of
$\mat{X}_{(n)}$ ordered lexicographically by index.

We denote the \emph{$n$-mode product} of a tensor
$\tensor{X} \in \R^{I_1 \times \dots \times I_N}$ and matrix
$\mat{A} \in \R^{J \times I_N}$ by
$\tensor{Y} = \tensor{X} \times_n \mat{A}$, where
$\tensor{Y} \in \R^{I_1 \times \dots \times I_{n-1} \times J \times I_{n+1} \times \dots \times I_{N}}$.
This operation multiplies each mode-$n$ fiber of $\tensor{X}$ by $\mat{A}$, and can be expressed element-wise as
\[
    (\tensor{X} \times_{n} \mat{A})_{i_1 \dots i_{n-1} j i_{n+1} \dots i_{N}}
    =
    \sum_{i_n=1}^{I_n} x_{i_1 i_2 \dots i_N} a_{j i_n}.
\]
The inner product of two tensors
$\tensor{X}, \tensor{Y} \in \R^{I_1 \times \dots \times I_N}$
is the sum of the products of their entries:
\[
    \inner{\tensor{X}, \tensor{Y}}
    =
    \sum_{i_1=1}^{I_1}
    \sum_{i_2=1}^{I_2}
    \dots
    \sum_{i_N=1}^{I_N}
    x_{i_1 i_2 \dots i_N}
    y_{i_1 i_2 \dots i_N}.
\]
The Frobenius norm of a tensor $\tensor{X}$ is 
$\norm{\tensor{X}}_{\frobenius} = \sqrt{\inner{\tensor{X}, \tensor{X}}}$.

\paragraph{Tucker decomposition.}
The \emph{Tucker decomposition} of a tensor
$\tensor{X} \in \R^{I_1 \times \dots \times I_N}$
decomposes $\tensor{X}$ into a
\emph{core tensor} $\tensor{G} \in \R^{R_1 \times \dots \times R_N}$
and $N$ \emph{factor matrices}
$\mat{A}^{(n)} \in \R^{I_n \times R_n}$.
We refer to $\mat{r} = (R_1,\dots,R_N)$ as the \emph{core shape},
which is also called the \emph{multilinear rank} or \emph{truncation}
of the decomposition.
We denote the loss of an optimal rank-$\mat{r}$ Tucker decomposition by
\[
    L(\tensor{X}, \mat{r})
    \defeq
    \min_{\tensor{G} \in \R^{R_1 \times \dots \times R_N}}
    \norm{\tensor{X} - \tensor{G} \times_{1} \mat{A}^{(1)} \times_{2} \dots \times_{N} \mat{A}^{(N)}}_{\frobenius}^2.
\]

\vspace{-0.15cm}

\section{Reduction to HOSVD Tucker packing}
\label{sec:hosvd_packing}

\subsection{Higher-order singular value decomposition}
\label{subsec:hosvd}

We start with a recap of the seminal work 
on higher-order singular value decompositions (HOSVD)
by~\citet*{HOSVD}.

\begin{theorem}[{\citet[Theorem 2]{HOSVD}}]
\label{thm:hosvd}
Any tensor $\tensor{X}\in\mathbb{R}^{I_1\times\cdots\times I_N}$
can be written as
\begin{align*}
    \tensor{X} = \tensor{S} \times_1 \mat{U}^{(1)} \times_2 \cdots \times_N \mat{U}^{(N)},
\end{align*}
where each $\mat{U}^{(n)} \in {\R}^{I_n \times I_n}$ is an orthogonal matrix
and
$\tensor{S} \in \mathbb{R}^{I_1\times\cdots\times I_N}$ is a tensor
with subtensors $\tensor{S}_{i_n = \alpha}$, obtained by fixing the
$n$-th index to $\alpha$, that have the properties:
\begin{enumerate}
    \item \emph{all-orthogonality:}
    for all possible values of $n$, $\alpha$ and $\beta$ subject to $\alpha \ne \beta$,
      two subtensors $\tensor{S}_{i_n = \alpha}$ and $\tensor{S}_{i_n = \beta}$
      are orthgonal, i.e.,
      $
        \inner*{\tensor{S}_{i_n = \alpha}, \tensor{S}_{i_n = \beta}} = 0 \text{~when~} \alpha \ne \beta;
      $
    \item \emph{ordering:}
    for all values of $n$,
    $
        \norm{\tensor{S}_{i_n=1}}_{\frobenius}
        \ge
        \norm{\tensor{S}_{i_n=2}}_{\frobenius}
        \ge
        \dots
        \ge
        \norm{\tensor{S}_{i_n=I_n}}_{\frobenius}
        \ge 0.
    $
\end{enumerate}
Furthermore, the values $\norm{\tensor{S}_{i_n=i}}_{\frobenius}$,
denoted by $\sigma_{i}^{(n)}$,
are the singular values of the mode-$n$ unfolding $\mat{X}_{(n)}$,
and the columns of
$\mat{U}^{(n)}$ are the the left singular vectors.
\end{theorem}

Next, we present the $\texttt{TuckerHOSVD}$ algorithm.
This is a widely used initialization strategy
when computing rank-$\mat{r}$
Tucker decompositions~\citep{kolda2009tensor},
i.e., if the core shape~$\mat{r}$ is predetermined.

\begin{algorithm}[H]
\caption{$\texttt{TuckerHOSVD}$}
\label{alg:tucker_hosvd}
\textbf{Input:} $\tensor{X} \in \R^{I_1 \times \cdots \times I_N}$,
 core shape $\mat{r} = (R_1,\dots,R_N)$

\begin{algorithmic}[1]
    \FOR{$n=1$ to $N$}
        \STATE $\mat{A}^{(n)} \gets R_n \text{ top left singular vectors of } \mat{X}_{(n)}$
    \ENDFOR
    \STATE $\tensor{G} \gets \tensor{X} \times_{1} {\mat{A}^{(1)}} \times_{2} \dots \times_{N} {\mat{A}^{(N)}}$
    \STATE \textbf{return}
      $\tensor{G}, \mat{A}^{(1)}, \mat{A}^{(2)}, \dots, \mat{A}^{(N)}$ 
\end{algorithmic}
\end{algorithm}

The output of \texttt{TuckerHOSVD} has the following error guarantees~\citep{HOSVD, hackbusch2019tensor}.
These bounds suggest a less expensive
\emph{surrogate loss function} to minimize instead
when optimizing the core tensor
shape subject to a Tucker decomposition size constraint.
We give self-contained proofs of these results
in \Cref{subsec:reconstruction_error_proofs}
that build on \Cref{thm:hosvd}.

\begin{restatable}[{\citet[Property 10]{HOSVD};
\citet[Theorem 10.2]{hackbusch2019tensor}}]{theorem}{HosvdTruncationError}
\label{thm:tensor_subspace_approximation}
For any tensor $\tensor{X} \in \R^{I_1 \times \dots \times I_N}$ and core shape $\mat{r} \in [I_1] \times \dots \times [I_N]$,
let the output of $\textup{\texttt{TuckerHOSVD}}(\tensor{X}, \mat{r})$ be
$\tensor{G} \in \R^{R_1 \times \dots \times R_N}$ and
$\mat{A}^{(n)} \in \R^{I_n \times R_n}$, for each $n \in [N]$.
If we let
\begin{equation}
\label{def:hosvd_reconstruction}
    \widehat{\tensor{X}}_{\textnormal{HOSVD}(\mat{r})}
    \defeq
    \tensor{G} \times_{1} \mat{A}^{(1)} \times_{2} \dots \times_{N} \mat{A}^{(N)}
\end{equation}
denote the reconstructed $\mat{r}$-truncated tensor, then
\begin{align*}
    \norm{\tensor{X} - \widehat{\tensor{X}}_{\textnormal{HOSVD}(\mat{r})}}_{\frobenius}^2
    &\le \sum_{n=1}^N \sum_{i_n = R_n + 1}^{I_n} \parens*{\sigma_{i_n}^{(n)}}^2 \\
    &\le N \cdot \Loss(\tensor{X}, \mat{r}).
\end{align*}
Furthermore, we have
$
    \Loss(\tensor{X}, \mat{r}) \le \norm{\tensor{X} - \widehat{\tensor{X}}_{\textnormal{HOSVD}(\mat{r})}}_{\frobenius}^2.
$
\end{restatable}
\cref{thm:tensor_subspace_approximation} implies that the following function is a meaningful proxy for the reconstruction error of an optimal rank-$\mat{r}$ Tucker decomposition.

\begin{definition}
Define the \emph{surrogate loss} of core shape $\mat{r}$ as
\begin{align}
\label{eqn:surrogate_loss_def}
    \SurrogateLoss(\tensor{X}, \mat{r})
    &\defeq
    \sum_{n=1}^N \sum_{i_n = R_n + 1}^{I_n} \parens*{\sigma_{i_n}^{(n)}}^2.
\end{align}
\end{definition}

To summarize so far, 
for any core shape $\mat{r} \in [I_1] \times \dots \times [I_N]$,
we are guaranteed that
$
  \sfrac{1}{N} \cdot \SurrogateLoss(\tensor{X}, \mat{r}) \le
  \Loss(\tensor{X}, \mat{r}) \le \SurrogateLoss(\tensor{X}, \mat{r}).
$
We refer the reader to
\Cref{subsec:proxy_sandwhich} for the full details.

\subsection{Tucker packing problem}

Next, observe that the sum of squared singular values
across all mode-$n$ unfoldings of $\tensor{X}$
is
\begin{align*}
    \sum_{n=1}^N \sum_{i_n = 1}^{I_n} \parens*{\sigma_{i_n}^{(n)}}^2
    =
    \sum_{n=1}^N \norm{\mat{X}_{(n)}}_{\frobenius}^2
    =
    N \norm{\tensor{X}}_{\frobenius}^2.
\end{align*}
This means we can solve a singular value packing problem
instead by considering the complement of the surrogate loss.
The following lemma is a wrapper
for the truncated HOSVD error guarantees in \Cref{thm:tensor_subspace_approximation}.

\begin{restatable}[]{lemma}{HosvdPackingProblem}
\label{lem:hosvd_packing_problem}
For any tensor $\tensor{X} \in \R^{I_1 \times \dots \times I_N}$
and budget $c \ge 1 + \sum_{n=1}^N I_n$ for the size of the Tucker decomposition,
let the set of feasible core shapes be
\[
    F = \set*{\mat{r} \in [I_1] \times \cdots \times [I_N] : \prod_{n=1}^N R_n + \sum_{n=1}^N I_n R_n \le c}.
\]
Then, we have
\begin{align*}
    \widetilde{\mat{r}}^*
    \defeq
    \argmin_{\mat{r} \in F} \SurrogateLoss(\tensor{X}, \mat{r})
    =
    \argmax_{\mat{r} \in F} \sum_{n=1}^N \sum_{i_n=1}^{R_n} \parens*{\sigma_{i_n}^{(n)}}^2.
\end{align*}
Further, if 
$\mat{r}^* \defeq \argmin_{\mat{r} \in F} L(\tensor{X},\mat{r})$
is an optimal budget-constrained core shape, then
$
    \Loss\parens*{\tensor{X},\widetilde{\mat{r}}^*}
    \le
    N \cdot
    \Loss\parens*{\tensor{X},\mat{r}^*}.
$
\end{restatable}

To find a core shape whose optimal Tucker decomposition
approximates the optimal loss
$L(\tensor{X}, \mat{r}^*)$ subject to a size  constraint,
we solve the maximization problem
in \Cref{lem:hosvd_packing_problem}.
Optimizing this proxy objective
is substantially less expensive than methods that rely on
rank-$\mat{r}$ Tucker decomposition solvers as a subroutine.
We formalize this idea by introducing the more general problem below.

\begin{definition}[Tucker packing problem]
Given a shape $(I_1, \dots, I_N) \in \Z_{\ge 1}^N$,
$N$ non-increasing sequences
$\mat{a}^{(n)} \in \R_{\ge 0}^{I_n}$,
and a budget $c \ge 1$,
the \emph{Tucker packing problem}
asks to find a core shape
$\mat{r} \in [I_1] \times \dots \times [I_N]$ that solves:
\begin{tcolorbox}[enhanced, opacityframe=1, colback=white!99!black, top=-0.8mm]
\begin{align}
   &\text{maximize~~} \sum_{n=1}^N \sum_{i_n=1}^{R_n} a_{i_n}^{(n)} \label{eq:tucker-packing} \\
   &\text{subject to~~} \prod_{n=1}^N R_n + \sum_{n=1}^N I_n R_n \le c \label{eq:tucker-packing-constraint}
\end{align}
\end{tcolorbox}
\end{definition}

We also denote the objective by
$
    f(\mat{r}) \defeq \sum_{n=1}^N \sum_{i_n=1}^{R_n} a_{i_n}^{(n)}.
$

\begin{restatable}[]{theorem}{Hardness}
\label{thm:hardness}
The Tucker packing problem is \textnormal{NP-hard}.
\end{restatable}

We prove this result in~\Cref{app:hardness}
with an intricate reduction from the EQUIPARTITION problem
(see, e.g., \citet[SP12]{garey1979computers}).

NP-hardness motivates the need for efficient approximation algorithms.
In \Cref{sec:algorithm}, we develop a
\emph{polynomial-time approximation scheme} (PTAS) for the Tucker packing problem.
We leave the existence of a fully-polynomial time approximation scheme (FPTAS)
as a challenging open question for future works.

To conclude,
since Tucker packing is the complement of surrogate loss minimization,
we must quantify how a $(1-\varepsilon)$-approximation for the packing problem
can affect the error incurred in the surrogate loss.
We explain this in detail in \Cref{app:complement_approximation_error}
and present the main idea below.

\begin{restatable}[]{lemma}{ComplementApproximationError}
\label{lem:complement_approximation_error}
Let $\mat{r} \in [I_1] \times \dots \times [I_N]$ be any core shape
that achieves a $(1-\varepsilon / N)$-approximation to the Tucker packing problem.
Then, we have
$
    \textnormal{RRE}(\tensor{X}, \mat{r})
    \le
    N \cdot \textnormal{RRE}(\tensor{X}, \mat{r}^*) + \varepsilon,
$    
where $\textnormal{RRE}(\tensor{X}, \mat{r}) := \Loss(\tensor{X}, \mat{r})/\norm{\tensor{X}}_{\frobenius}^2$.
\end{restatable}

\begin{remark}
We can obtain global approximation guarantees for Tucker decomposition reconstruction error
by (1) finding an approximately optimal core shape,
(2) running \texttt{TuckerHOSVD} to initialize the Tucker decomposition,
and (3) using alternating least squares (ALS) to improve the tensor decomposition.
This is analogous to how $k$-means++ enhances Lloyd's algorithm~\citep{arthur2006k}.
\end{remark}
\section{Algorithm}
\label{sec:algorithm}

\subsection{Warm-up:\hspace{-0.02cm} Connections to multiple-choice knapsack}

To start, consider a simplified version of the Tucker packing problem
that only accounts for the size of the core tensor,
i.e., the factor matrices do not use any of the budget.
We show that after two simple transformations this new problem reduces
to the \emph{multiple-choice knapsack problem},\footnote{The multiple-choice knapsack problem is a 0-1 knapsack problem where the items are
partitioned into $N$ classes and exactly one item must be taken from each class~\cite{sinha1979multiple}.} which is NP-hard~\citep{pisinger1995minimal}
but has an FPTAS~\citep{lawler1977fast}.

Concretely, the optimization problem is
\begin{align}
   \text{maximize~~} &\sum_{n=1}^N \sum_{i_n=1}^{R_n} a_{i_n}^{(n)} \label{eq:tucker-packing-simplified} \\
   \text{subject to~~} &\prod_{n=1}^N R_n \le c \label{eq:tucker-packing-constraint-simplified}
\end{align}

\paragraph{Prefix sums transformation.}
To get closer to a {0-1} knapsack problem, define new coefficients
by taking the prefix sums of the $a_{i_n}^{(n)}$'s, for each
$n \in [N]$ and $i_n \in [I_n]$:
\begin{align*}
    p_{i_n}^{(n)} \defeq \sum_{j_n=1}^{i_n} a_{j_n}^{(n)}.
\end{align*}
This ``core size-only'' Tucker packing problem can
be reformulated as the following integer program:
\begin{align}
   \text{maximize~~} & \sum_{n=1}^N \sum_{i_n=1}^{I_n} p_{i_n}^{(n)} x_{i_n}^{(n)} \notag \\
   \text{subject to~~} & \prod_{n=1}^N \sum_{i_n=1}^{I_n} i_n x_{i_n}^{(n)} \le c \label{eq:tucker-packing-const-conc-before} \\
   &
   \sum_{i_n=1}^{I_n} x_{i_n}^{(n)} = 1
   \hspace{1cm} \forall n \in [N] \notag \\
   & x_{i_n}^{(n)} \in \{0, 1\} \hspace{1.11cm}\forall n\in[N], i_n \in [I_n] \notag
\end{align}
We optimize over $i_n$ instead of $R_n$ for notational brevity.

\paragraph{Log transformation.}
Next, replace constraint~\eqref{eq:tucker-packing-const-conc-before} with the linear inequality
\begin{align*}
    \sum_{n=1}^N \sum_{i_n=1}^{I_n} \log(i_n) x_{i_n}^{(n)} \le \log(c).
\end{align*}
This substitution is valid because in any feasible solution,
for each $n\in[N]$, exactly one of $x_1^{(n)}, x_2^{(n)}, \dots,x_{I_n}^{(n)}$ is equal to one and the rest are zero.
Putting everything together,
this core size-only Tucker packing problem is the following multiple-choice knapsack problem:
\begin{align}
   \text{maximize~~} & \sum_{n=1}^N \sum_{i_n=1}^{I_n} p_{i_n}^{(n)} x_{i_n}^{(n)} \label{eq:tucker-packing-int-prog-obj} \\
   \text{subject to~~} & \sum_{n=1}^N \sum_{i_n=1}^{I_n} \log(i_n) x_{i_n}^{(n)} \le \log(c) \notag \\ 
   &
   \sum_{i_n=1}^{I_n} x_{i_n}^{(n)} = 1 \hspace{1cm} \forall n \in [N]
   \notag \\
   &
   x_{i_n}^{(n)} \in \{0, 1\} \hspace{1.11cm} \forall n\in[N], i_n \in [I_n] 
   \notag 
\end{align}

\begin{restatable}[\citealt{lawler1977fast}]{theorem}{multiChoiceKnapsack}
\label{thm:multi-choice-knapsack}
There exists an algorithm that computes a $(1-\epsilon)$-approximation to problem \eqref{eq:tucker-packing-int-prog-obj} in time and space
$O(N^2 \varepsilon^{-1} \sum_{n=1}^N I_n )$.
\end{restatable}

The FPTAS in~\Cref{thm:multi-choice-knapsack}
for multiple-choice knapsack
uniformly downscales all coefficients $p_{i_n}^{(n)}$ in the objective,
rounds them,
and then uses dynamic programming.

\subsection{PTAS for the Tucker packing problem}

Now we consider the true cost of a Tucker decomposition,
i.e., the size of the core tensor and the factor matrices.
We first introduce a simple grid-search algorithm
that solves approximate Tucker packing for a
general class of feasible solutions (i.e., \emph{downwards closed sets}).
This captures the Tucker packing problem
and will be useful for extending our results to tree tensor networks in~\Cref{sec:htucker}.

\begin{definition}
For any $N \ge 1$ and $(I_1, \dots, I_N) \in \Z_{\ge 1}^N$,
let $F \subseteq [I_1] \times \dots \times [I_N]$.
The set $F$ is \emph{downward closed} if for any pair $(R_1,\dots,R_N),(R'_1,\dots,R'_N) \in [I_1] \times \dots \times [I_N]$ such that $R'_n \leq R_n$ for all $n\in[N]$, $(R_1,\ldots,R_N)\in F$ implies that $(R'_1,\ldots,R'_N)\in F$.
\end{definition}

\begin{restatable}{lemma}{GridSearch}
\label{lem:grid-search}
Let $0 < \varepsilon \le 1$ and $F \subseteq [I_1] \times \dots \times [I_N]$ be downwards closed.
For each $n\in[N]$, define
\[
    S_n^{(\varepsilon)} = \set*{\ceil{(1+\epsilon)^k}:k\in\mathbb{Z}_{\geq 0}, \ceil{(1+\epsilon)^k} \leq I_n}.
\]
Let $\mat{r}^*$ be an optimal solution to the generalized problem 
\begin{align}
\label{eq:comb_opt_whole_rep_prob}
\textnormal{maximize~~} & \sum_{n=1}^N \sum_{i_n=1}^{R_n} a_{i_n}^{(n)} \\ \nonumber
   \textnormal{subject to~~} & (R_1,\ldots,R_N)\in F 
\end{align}
and let $\mat{r}^{(\varepsilon)} = (R^{(\varepsilon)}_1,\ldots,R^{(\varepsilon)}_N)$ be an optimal solution to
\begin{align}
\label{eq:comb_opt_whole_rep_grid_prob}
\textnormal{maximize~~} & \sum_{n=1}^N \sum_{i_n=1}^{R_n} a_{i_n}^{(n)} \\
\nonumber
   \textnormal{subject to~~} & (R_1,\ldots,R_N)\in (S^{(\varepsilon)}_1 \times \dots \times S^{(\varepsilon)}_N) \cap F \nonumber
\end{align}
Then, $f(\mat{r}^{(\varepsilon)}) \ge (1 + \varepsilon)^{-1} f(\mat{r}^*)$.
Further, there is an algorithm that finds an optimal solution of \eqref{eq:comb_opt_whole_rep_grid_prob} with running time
$
    O\parens{\sum_{n=1}^N I_n + \varepsilon^{-N} \prod_{n=1}^N \parens{1 + \log_{2}(I_n)}}.
$
\end{restatable}

The time complexity in~\Cref{lem:grid-search} is exponential in the order
of the tensor, so we now focus on designing our main algorithm
$\texttt{TuckerPackingSolver}$, whose running time is
$O(\text{poly}(N, \log c, I_n))$
for constant values of $\varepsilon > 0$.

\paragraph{Algorithm description.}
There are two phases in this algorithm:
(1) exhaustively search over all ``small'' core shapes
by trying all budget allocation splits when the factor matrix cost is low;
and
(2) try coarser splits between
the core tensor size and factor matrix costs.
In the large phase,
we show that it is sufficient to consider
$O(\log_{1+\varepsilon} c)$ such splits.
Each budget split induces a problem of the form \eqref{eq:one-large-dim-prob},
which after applying prefix sum and log transformations becomes a
familiar \emph{2-dimensional knapsack problem with a partition matroid constraint}:
\begin{align}
   \text{maximize~~} & \sum_{n=1}^N \sum_{i_n=1}^{I_n} p_{i_n}^{(n)} x_{i_n}^{(n)} \label{eq:tucker-packing-split-int-prog-obj} \\
   \text{subject to~~} & \sum_{n=1}^N \sum_{i_n=1}^{I_n} \log(i_n) x_{i_n}^{(n)} \le \log(c_{\text{core}}) \notag \\
   & \sum_{n=1}^N \sum_{i_n=1}^{I_n} I_n i_n x_{i_n}^{(n)} \le c - c_{\text{core}} \notag \\
   &
   \sum_{i_n=1}^{I_n} x_{i_n}^{(n)} = 1 \hspace{1cm} \forall n \in [N]
   \notag \\
   &
   x_{i_n}^{(n)} \in \{0, 1\} \hspace{1.11cm} \forall n\in[N], i_n \in [I_n] 
   \notag
\end{align}

More generally, \Cref{eq:tucker-packing-split-int-prog-obj} is a
\emph{$d$-budgeted matroid independent set problem}
(see, e.g., \citet{grandoni2014new}),
in which a linear objective function is maximized subject to $d$ knapsack constraints and a matroid constraint.
Recall that the multi-dimensional knapsack problem (even without any matroid constraints)
does not admit an FPTAS unless $\text{P} = \text{NP}$~\citep{gens1979computational, korte1981existence}.
It does, however, have a PTAS as shown by the next theorem.

\begin{theorem}[{\citealt[Corollary 4.4]{grandoni2014new}}]
\label{thm:grandoni2014new}
There is a PTAS (i.e., a $(1-\varepsilon)$-approximation algorithm) for
the $d$-budgeted matroid independent set problem with running time $O(m^{O(d^2/\varepsilon)})$,
where $m$ is the number of items.
\end{theorem}

The number of items in~\Cref{eq:tucker-packing-split-int-prog-obj} is $m=\sum_{n=1}^N I_n$,
one for each core shape dimension choice.
Thus, since $d=2$, this gives a running time of
$O\parens{\parens{\sum_{n=1}^N I_n}^{O(1/\varepsilon)}}$.
This in turn allows us to bound the overall
running time of~\cref{alg:total-size-Tucker}.

\begin{remark}
We can use integer linear programming solvers
for~\Cref{eq:tucker-packing-split-int-prog-obj} instead of~\citet{grandoni2014new},
but this is possible only because we decouple the core shape cost
and the factor matrix cost,
i.e., because of the $c_{\text{factor}}$ and $c_{\text{core}}$ budget splits.
\end{remark}

\begin{algorithm}[H]
   \caption{\texttt{TuckerPackingSolver}}
   \label{alg:total-size-Tucker}
   
   {\bfseries Input:}
   shape $(I_1,\dots,I_N) \in \mathbb{Z}_{\ge 1}^N$,
   $N$ non-increasing sequences $\mat{a}^{(n)}\in\R^{I_n}_{\geq 0}$, budget $c \ge 1 + \sum_{n=1}^N I_n$, error $\epsilon > 0$
   
\begin{algorithmic}[1]
    \STATE Initialize $S \gets \varnothing$
    \FOR{$c_{\text{factor}} \in [\ceil{1/\varepsilon} \sum_{n=1}^N I_n]$}
        \STATE Let $\mat{r}'$ be a $(1-\varepsilon)$-approximate solution to:
        \begin{align}
        \label{eq:tucker-packing-brute-force}
            \text{maximize~~} & \MaxObj(\mat{r}) \\ 
            \nonumber \text{subject to~~}
                & \textstyle \prod_{n=1}^N R_n \le c - c_{\text{factor}} \\
                \nonumber
                & \textstyle \sum_{n=1}^N I_n R_n \le c_{\text{factor}} \\
                \nonumber
                & R_n \in [\min(\ceil{1/\varepsilon}, I_n)] \hspace{1cm} \forall n\in [N]
        \end{align}
        \STATE Update $S \gets S \cup \{\mat{r}'\}$
    \ENDFOR

    \FOR{$k=0$ to $\floor{\log_{1+\epsilon} c}$}
        \STATE Set $c_{\text{core}} \gets (1+\epsilon)^k$   
        \STATE Let $\mat{r}^{(k)}$ be a $(1-\varepsilon)$-approximate solution to:
   \begin{align}
   \label{eq:one-large-dim-prob}
    \text{maximize~~} & \MaxObj(\mat{r}) \\ \nonumber
    \text{subject to~~} & \textstyle \prod_{n=1}^N R_n  \leq c_{\text{core}} \\ \nonumber
    & \textstyle \sum_{n=1}^N I_n R_n \leq c - c_{\text{core}} \\ \nonumber
    & R_n \in [I_n] \hspace{2.55cm} \forall n\in [N]
   \end{align}
   
   \STATE Update $S \gets S \cup \{\mat{r}^{(k)}\}$
   
   \ENDFOR
   \STATE \textbf{return} $\argmax_{\mat{r}\in S} \MaxObj(\mat{r})$
\end{algorithmic}
\end{algorithm}

\begin{theorem}
\label{thm:budget-splitting}
If $0 < \varepsilon < 1/3$,
then \Cref{alg:total-size-Tucker} returns a
$(1-3\epsilon)$-approximate solution
to Problem~\eqref{eq:comb_opt_whole_rep_prob} in time
$
    O\parens{\parens{\log_{1+\epsilon}(c) + \ceil*{1/\varepsilon}\sum_{n=1}^N I_n} \parens{\sum_{n=1}^N I_n}^{O\parens*{1/\varepsilon}}}.
$
\end{theorem}

\begin{proof}
\texttt{TuckerPackingSolver} solves for two types of shapes:
(1) ``small'' solutions where each $R_n \le \ceil{1/\varepsilon}$,
and (2) ``large'' solutions where $R_n > \ceil{1/\varepsilon}$ for some $n \in [N]$.

In the small phase,
observe that since $R_n \le \min(\ceil{1/\varepsilon}, I_n)$,
the factor matrix cost is
$\sum_{n=1}^N I_n R_n \le \ceil{1/\varepsilon}\sum_{n=1}^N I_n$.
Therefore, we can exhaustively check all small
budget splits of the form
$c_{\text{factor}} \in [\ceil{1/\varepsilon}\sum_{n=1}^N I_n]$.
Each split induces a 2-dimensional knapsack problem
with a partition matroid (but for a smaller set of items),
so use \Cref{thm:grandoni2014new}
to obtain a $(1-\varepsilon)$-approximation
for each subproblem.
If an optimal solution $\mat{r}^*$ to the Tucker packing problem
is small, then \Cref{alg:total-size-Tucker} 
recovers an approximately optimal objective.

For the large phase, assume that the optimal core shape
$\mat{r}^*$ has a dimension $m \in [N]$ such that
$R_{m}^* > \ceil{1/\varepsilon}$.
The algorithm searches over large shapes indirectly by splitting the budget $c$
between the size of the core tensor and the total cost of factor matrices.
A crucial observation is that we only need to check
$O(\log_{1+\varepsilon} c)$ different splits because \emph{there is a
sufficient amount of slack in the large dimension $R_{m}^*$}.

To proceed, let
$
    c^*_{\text{core}} = \prod_{n=1}^N R_n^*,
$
and let $k^*$ be the largest integer such that
$(1+\varepsilon)^k \le c^*_{\text{core}}$.
Define
$\widehat{c}_{\text{core}} = (1+\varepsilon)^{k^*}$,
and let
$\widehat{\mat{r}} = (\widehat{R}_1, \dots, \widehat{R}_N)$
where
\[
    \widehat{R}_n
    =
    \begin{cases}
        R_{n}^* & \text{if $n \ne m$}, \\
        \floor{R_{n}^* / (1 + \varepsilon)} & \text{if $n = m$}.
    \end{cases}
\]
Since we assumed $R_{m}^* \ge \ceil{1/\varepsilon}$ is a large dimension,
$
    \widehat{R}_{m}
    =
    \floor{R_m^* / (1+\varepsilon)}
    \ge
    \floor{\ceil{1/\varepsilon}/(1+\varepsilon)}.
$
Next, observe that
\begin{align*}
    \prod_{n=1}^N \widehat{R}_n
    &\le
    \frac{1}{1+\varepsilon} \prod_{n=1}^N R_{n}^*
    =
    \frac{c_{\text{core}}}{1+\varepsilon} \\
    &<
    \frac{(1+\varepsilon)^{k^* + 1}}{1+\varepsilon}
    =
    (1+\varepsilon)^{k^*}
    =
    \widehat{c}_{\text{core}},
\end{align*}
and
$
    \sum_{n=1}^N I_n \widehat{R}_n
    \le
    \sum_{n=1}^N I_n R^*_n
    \le
    c - c_{\text{core}}^*
    \le
    c - \widehat{c}_{\text{core}}.
$
Therefore, $\widehat{\mat{r}}$ is a feasible solution
of~\Cref{eq:one-large-dim-prob} for $k=k^*$.
Furthermore,
$f(\mat{r}^{(k^*)}) \ge (1-\varepsilon) f(\widehat{\mat{r}})$.

Next, we show that $f(\mat{\widehat{r}}) \geq (1-2\varepsilon) f(\mat{r}^*)$.
For any $n \neq m$, we have $\widehat{R}_n = R_n^*$, so it follows that
\begin{equation}
\label{eqn:final_proof_eqn_00}
    \sum_{i_n=1}^{\widehat{R}_n} a_{i_n}^{(n)}
    =
    \sum_{i_n=1}^{R_n^*} a_{i_n}^{(n)}.
\end{equation}
For the large dimension $m$, we have
\begin{align*}
    \widehat{R}_m
    =
    \floor*{\frac{R_{m}^*}{1 + \varepsilon}}
    \ge
    \frac{R_{m}^*}{1 + \varepsilon} - 1
    =
    \frac{R_{m}^* - (1 + \varepsilon)}{1 + \varepsilon}
    \cdot
    \frac{R_{m}^*}{R_{m}^*}.
\end{align*}
The assumption $R_{m}^* \ge \ceil{1/\varepsilon}$ then gives
\begin{align*}
    \frac{R_m^* - (1 + \varepsilon)}{R_m^*}
    \ge
    1 - \frac{1 + \varepsilon}{\ceil{1/\varepsilon}}
    \ge
    1 - \varepsilon - \varepsilon^2.
\end{align*}
It follows for any $\varepsilon \ge 0$ that
\[
    \widehat{R}_{m}
    \ge
    \frac{1-\varepsilon - \varepsilon^2}{1+\varepsilon} R_{m}^*
    \ge
    (1-2\varepsilon) R_{m}^*.
\]
Since $a_{i_1}^{(m)} \geq \dots \geq a_{I_m}^{(m)}$ is non-increasing, we have
\begin{equation}
\label{eqn:final_proof_eqn_01}
    \sum_{i_m=1}^{\widehat{R}_m} a_{i_m}^{(m)}
    \geq
    (1-2\varepsilon) \sum_{i_m=1}^{\ropti{m}} a_{i_m}^{(m)}.
\end{equation}
Combining \Cref{eqn:final_proof_eqn_00} and \Cref{eqn:final_proof_eqn_01}
gives $f(\widehat{\mat{r}}) \ge (1-2\varepsilon) f(\mat{r}^*)$,
so then
$
    f(\mat{r}^{(k^*)})
    \geq
    (1-\varepsilon) (1-2\varepsilon) f(\mat{r}^*)
    \geq
    (1-3\varepsilon)f(\mat{r}^*),
$
which proves the approximation guarantee.\footnote{
To obtain a PTAS, it is enough to consider the budget splits
$c_{\text{factor}} \in [\sum_{n=1}^N I_n^2]$,
i.e., line~2 of \texttt{TuckerPackingSolver},
but this gives a worse running time as explained in~\cref{app:algorithm}.}

Finally, the running time follows from our reductions to
the 2-dimensional knapsack problem with a partition matroid constraint,
and using \Cref{thm:grandoni2014new} for each budget split.
\end{proof}

\section{Tree tensor network decompositions}
\label{sec:htucker}

Here we consider a general decomposition called tree tensor network~\cite{oseledets2009breaking,kramer2020tree},
which includes Tucker decomposition, tensor-train decomposition, and hierarchical Tucker decomposition as special cases.

\begin{definition}[Tree tensor network]
\label{def:tree-tensor}
Let $\tensor{X}\in\R^{I_1\times \cdots\times I_N}$ be any tensor.
Let $G=(V,E)$ be a rooted tree with~$N$ leaves
where each node $v$ corresponds to a subset $S_v \subseteq [N]$.
The leaves are the $N$ singletons of $[N]$,
and internal nodes~$v$ are recursively defined by
$S_v = \cup_{u \in C_v} S_u$,
where $C_v$ is the set of children of $v$.

Each edge $e \in E$ is endowed an integer $R_e \ge 1$.
Then, a (truncated) \emph{tree tensor network} decomposition
of $\tensor{X}$ for tree $G$ is the following collection of tensors,
each corresponding to a $v \in V$.
For each leaf, the tensor is $\mat{A}^{(v)} \in \R^{I_n \times R_e}$,
where $e$ is the edge connecting $v$ to its parent.
For each internal node $v$, its tensor is
$\tensor{T}^{(v)} \in \R^{R_{e_1} \times \dots \times R_{e_k}}$,
where $E_v = \{e_1, \dots, e_k\}$ is the set of edges incident to $v$.


The output tensor $\widehat{\tensor{X}}$ is constructed by taking 
the mode-wise products of all the ``node tensors'' in $G$ over their corresponding edges.
These products commute and are associative
(see, e.g., Proposition 2.17 in~\citet{kramer2020tree}).

\end{definition}

\begin{remark}
The tree tensor network for Tucker decomposition corresponds to a
tree of depth one,
and for hierarchical Tucker decomposition
it is an (almost) balanced binary tree~\citep{oseledets2009breaking, grasedyck2010hierarchical}.
\end{remark}

We give an example with figures in~\Cref{app:tree-tensor-networks}.
Now we generalize the definition of tensor unfolding.

\begin{definition}
For any $\tensor{X}\in\R^{I_1\times \cdots\times I_N}$ and $S \subseteq [N]$,
the \emph{matricization} $\mat{X}_{(S)} \in \R^{P \times Q}$,
where $P=\prod_{n\in S} I_n$ and $Q=\prod_{n\in[N]\setminus S} I_n$,
is the matrix with the entries of $\tensor{X}$ arranged lexicographically
by their original index tuples.
\end{definition}

The next theorem gives a polynomial-time algorithm
for finding a tree tensor network decomposition that
achieves bounded reconstruction error for specified $R_e$ values.

\begin{theorem}[\citet{grasedyck2010hierarchical,kramer2020tree}]
Let $\tensor{X}\in\R^{I_1\times \cdots\times I_N}$
and $G=(V,E)$, $R_e$ for $e \in E$
be the tree tensor network parameters as in~\Cref{def:tree-tensor}.
There exists a polynomial-time algorithm that finds $\tensor{T}^{(v)}$
for $v\in V$ (for leaves these tensors are the matrices $\mat{A}^{(v)}$)
such that
\begin{align*}
    \norm{\tensor{X} - \widehat{\tensor{X}}}_{\frobenius}^2
    &\leq
    \sum_{v\in V \setminus \{r\}} \sum_{i_v = R_v + 1}^{P_v} \parens*{\sigma_{i_v}^{(v)}}^2 \\
    &\leq
     (\abs{V} - 1)\cdot \norm{\tensor{X} - \tensor{X}_{\textnormal{best}}}_{\frobenius}^2,
\end{align*}
where $\tensor{X}_{\textnormal{best}}$ is the best
tree tensor network decomposition for $G$ and the values $R_e$,
$r$ is the root node, $P_v = \prod_{n\in S_v} I_n$,
and $\sigma_{i}^{(v)}$ is the $i$-th singular value of $\mat{X}_{(S_v)}$.

Further, the size of the tree tensor network decomposition is
$
    \sum_{v\in L} I_v R_v + \sum_{v\in J} \prod_{e\in E_v} R_e,
$
where $L$ is the set of leaves,
$J$ is the set of internal nodes,
and $R_v$ is value on the edge that connects $v$ to its parent.
\end{theorem}

It follows that we can define the NP-hard
tree tensor network packing problem,
which generalizes Tucker packing.

\begin{definition}[Tree tensor network packing]
Given a shape $(I_1, \dots, I_N) \in \Z_{\ge 1}^N$,
tree $G=(V,E)$ with leaves $L=[N]$ and internal nodes $J$, 
$\abs{V}-1$ non-increasing sequences
(corresponding to non-root nodes)
$\mat{a}^{(v)} \in \R_{\ge 0}^{P_v}$ with $P_v=\prod_{n\in S_v} I_n$,
and the budget $c \ge 1$,
the \emph{tree tensor network packing} problem
asks to find $R_v \in \mathbb{Z}_{\ge 1}$ for $v\in V\setminus \{r\}$,
where $r$ is the root and $N_v$ is the set of neighbors of node $v$, that solves:
\vspace{-0.1cm}
\begin{align}
   &\text{maximize~} \sum_{v\in V\setminus \{r\}} \sum_{i=1}^{R_v} a_{i}^{(v)} \label{eq:tree-tucker-packing} \\
   &\text{subject to~} \hspace{0.31cm} \sum_{v\in J}\prod_{u\in N_v} R_u + \sum_{v\in L} I_v R_v \le c \label{eq:tree-tucker-packing-constraint}
\end{align}
\end{definition}

\begin{theorem}
\label{thm:tree-tensor-packing}
There is a $(1-\varepsilon$)-approximation algorithm for
the tree tensor network packing problem the runs in time
$
    O(\sum_{v\in V \setminus\{t\}} P_v +
    \varepsilon^{-(|V|-1)}\prod_{v\in V \setminus\{t\}} (1 + \log_{2}(P_v))).
$
\end{theorem}

\section{Experiments}
\label{sec:experiments}

\begin{figure*}
\centering
    \includegraphics[width=1.0\linewidth]{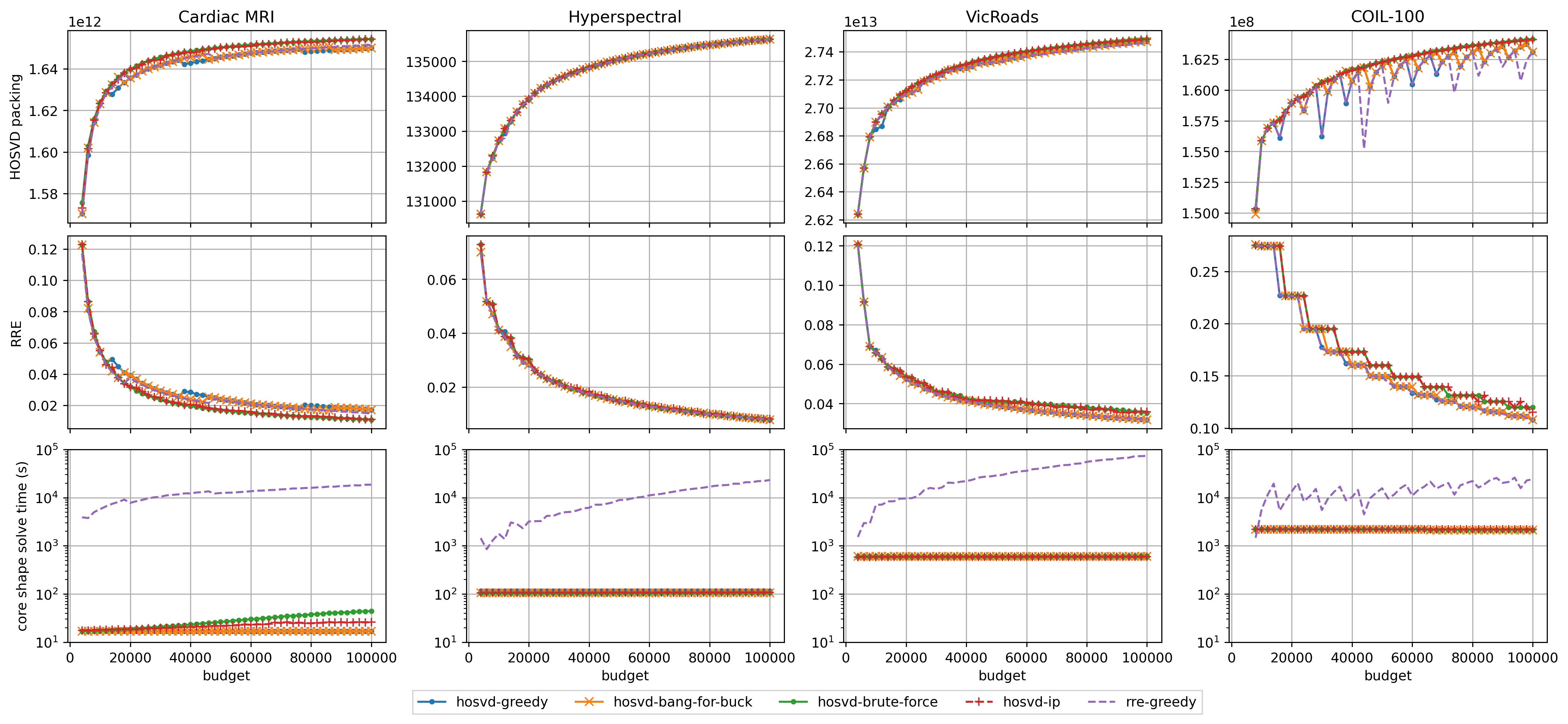}
\vspace{-0.75cm}
\caption{Comparison of five core shape solvers on four real-world tensors (columns) for increasing values of the Tucker decomposition size budget $c \le 100,000$.
The plots in the top row are the HOSVD Tucker packing objective value $f(\mat{r})$ for the core shape solutions $\mat{r}$,
the middle row is the RRE,
and the bottom row is the running time of each algorithm in seconds.}
\label{fig:solver-comparison}
\vspace{-0.3cm}
\end{figure*}

We compare several algorithms for computing
core shapes for size-constrained Tucker decompositions
on four real-world tensors (see \Cref{table:tensor-datasets}).
These experiments demonstrate the effectiveness of using
the surrogate loss $\widetilde{L}(\tensor{X},\mat{r})$ in place of
the true relative reconstruction error (RRE),
both in terms of solution quality and running time.
All of the experiments use NumPy~\citep{harris2020array}
with an Intel Xeon W-1235 processor (3.7 GHz, 8.25MB cache) and 128GB of RAM.

\begin{table}[h!]
\vspace{-0.25cm}
\caption{Statistics for tensor datasets used in experiments.}
\label{table:tensor-datasets}
\vspace{-0.25cm}
\begin{center}
\begin{small}
\begin{sc}
\begin{tabular}{lcr}
\toprule
tensor & shape & size \\
\midrule
cardiac mri & $256 \times 256 \times 14 \times 20$ & $18,350,080$ \\
hyperspectral & $1024 \times 1344 \times 33$ & $45,416,488$ \\
vicroads & $1084 \times 2033 \times 96$ & $211,562,112$ \\
coil-100 & $7200 \times 128 \times 128 \times 3$ & $353,894,400$ \\
\bottomrule
\end{tabular}
\end{sc}
\end{small}
\end{center}
\vspace{-0.35cm}
\end{table}

\subsection{Algorithms}

The first four algorithms we consider are based on HOSVD Tucker packing:
they compute the mode-$n$ singular values $\sigma_{i_n}^{(n)}$
and take $\mat{a}^{(n)} = \{(\sigma_{i_n}^{(n)})^2\}_{i_n=1}^{I_n}$
as input to their Tucker packing instance.
The fifth is a commonly used greedy algorithm
that computes true losses
$L(\tensor{X}, \mat{r})$ at each step.

\textbf{HOSVD-IP} is the \texttt{TuckerPackingSolver} algorithm with $\varepsilon=0.25$,
but we use the integer programming solver in \texttt{scipy.optimize.mlip} to solve
each budget split subproblem instead of the PTAS in~\citet{grandoni2014new}.

\textbf{HOSVD-greedy} maximizes the same packing objective
by
repeating
$\mat{r} \gets \argmax_{\mat{r}' \in N(\mat{r})} f(\mat{r}')$,
where $N(\mat{r})$ is the set of neighboring feasible core shapes
$\mat{r'} = \mat{r} + \mat{e}_{n}$
and $\mat{e}_n \in \{0,1\}^{N}$ is a standard unit vector.
This is Algorithm 3.1 in~\citet{ehrlacher2021adaptive}
with additional budget constraints.

\textbf{HOSVD-bang-for-buck} is analogous to HOSVD-greedy, but
it increments the dimension in each step that maximizes
$(f(\mat{r}') - f(\mat{r})) / (\text{cost}(\mat{r}') - \text{cost}(\mat{r}))$, where $\text{cost}(\mat{r})$
is the size of the rank-$\mat{r}$ Tucker decomposition.

\textbf{HOSVD-brute-force} exhaustively checks all feasible core shapes
and outputs the maximum Tucker packing objective.

\textbf{RRE-greedy} constructs the core shape by
computing $O(N)$ rank-$\mat{r}'$ Tucker decompositions in each step
and incrementing the dimension that most improves the RRE,
Concretely, the update is
$\mat{r} \gets \argmin_{\mat{r}' \in N(\mat{r})} L(\tensor{X}, \mat{r'})$,~similar to the Greedy-TL algorithm of~\citet{hashemizadeh2020adaptive}.

\subsection{Results}

We consider the budgets $c \le 100,000$ for all tensor datasets.
For each $c$, we run each algorithm to get core shape $\mat{r}$.
Then in \Cref{fig:solver-comparison},
we plot the packing objective $f(\mat{r})$,
the RRE, i.e., $L(\tensor{X},\mat{r}) / \norm{\tensor{X}}_{\frobenius}^2$,
and the algorithm running time
(including the mode-$n$ singular value computations) as a function of~$c$.
Each $L(\tensor{X}, \mat{r})$ computation uses
20 iterations of HOOI.

\textbf{Cardiac MRI}
shows that maximizing the HOSVD Tucker packing objective $f(\mat{r})$
can give noticeably better RRE than RRE-greedy,
while also running 1000x faster.
If we take a closer look at the core shapes these algorithms output,
HOSVD-\{brute-force, IP\} always return core shapes of the form $(x,y,z,1)$,
whereas the greedy algorithms allocate budget
to the time dimension as $c$ increases, e.g., $(x,y,z,3)$.
Increases in the fourth dimension correspond to points of
degradation in $f(\mat{r})$ and RRE in \Cref{fig:solver-comparison}.

This tensor is also small enough
to see differences in the running times of the HOSVD packing algorithms.
In particular, we see (1) that there is a fixed
cost for computing the $\sigma_{i_n}^{(n)}$'s,
and (2) that HOSVD-\{greedy, bang-for-buck\} are faster than
HOSVD-IP, which is faster than HOSVD-brute-force.
All algorithms are significantly faster than RRE-greedy.

\textbf{Hyperspectral} 
shows that the surrogate loss $\widetilde{L}(\tensor{X}, \mat{r})$
and RRE can guide greedy algorithms to the same core shapes,
and that HOSVD-greedy can achieve maximum the Tucker packing objective.
We see that computing the higher-order singular
values becomes the bottleneck for the HOSVD solvers,
not solving the packing instances themselves.

\textbf{VicRoads} 
shows a clear gap between RRE and the surrogate loss.
While HOSVD-\{greedy, bang-for-buck\} are suboptimal
in the packing objective,
they achieve the same RRE as the 100x slower RRE-greedy algorithm.
This data demonstrates a shortcoming of the surrogate loss,
but also shows that higher-order singular values can still be effective.

\textbf{COIL-100}
is perhaps the most interesting tensor
because it shows
the \emph{non-monotonic} behavior of greedy
HOSVD Tucker packing algorithms.
Similar to cardiac MRI,
every time a greedy core shape solver increases the
dimension of the first index (corresponding to the number of objects),
the packing objective $f(\mat{r})$.
This effect also appears in the RRE plots,
but it happens in the opposite direction.

\section*{Acknowledgements}
We thank the anonymous reviewers of an earlier version of this
work for pointing us to the multiple-choice knapsack literature.
We also want to thank Mohit Singh and Anupam Gupta for
fruitful discussions about our main algorithm.

\bibliography{references}
\bibliographystyle{icml2022}

\newpage
\appendix
\onecolumn

\section{Missing analysis from \Cref{sec:hosvd_packing}}
\label{app:hosvd_packing}

\subsection{Proof of \Cref{thm:tensor_subspace_approximation}}
\label{subsec:reconstruction_error_proofs}

\subsubsection{Upper bounding the reconstruction error}

For any Tucker decomposition
$\tensor{X} = \tensor{G} \times_1 \mat{A}^{(1)} \times_2 \cdots \times_N \mat{A}^{(N)} \in \R^{I_1 \times \dots \times I_N}$,
the mode-$n$ unfolding of~$\tensor{X}$ can
be written as
\begin{equation}
\label{eq:mode_n_mat}
    \mat{X}_{(n)}
    =
    \mat{A}^{(n)} \mat{G}_{(n)}
    \parens*{
      \mat{A}^{(1)} \ktimes \cdots \mat{A}^{(n-1)} \ktimes \mat{A}^{(n+1)} \ktimes \cdots \ktimes \mat{A}^{(N)}
    }^\intercal.
\end{equation}
This is the corrected version of Equation~(4.2) in \citet{kolda2009tensor}.

\begin{lemma}[{\citet[Property 10]{HOSVD}}]
\label{lemma:upper_bound_for_loss}
For any
$\tensor{X}\in\mathbb{R}^{I_1\times\cdots\times I_N}$
and $n \in [N]$,
let
$\sigma^{(n)}_{1} \geq \sigma^{(n)}_{2} \dots \geq \sigma^{(n)}_{I_n}$ denote the singular values of $\mat{X}_{(n)}$.
Then, for any core shape
$\mat{r} = (R_1, R_2, \dots, R_N) \in [I_1] \times [I_2] \times \dots \times [I_N]$, we have
\begin{align*}
    \Loss(\tensor{X}, \mat{r})
    &\le
    \norm{\tensor{X} - \widehat{\tensor{X}}_{\textnormal{HOSVD}(\mat{r})}}_{\frobenius}^2 \\
    &\le
    \sum_{i_1 = R_1 + 1}^{I_1} \parens*{\sigma_{i_1}^{(1)}}^2
    +
    \sum_{i_2 = R_2 + 1}^{I_2} \parens*{\sigma_{i_2}^{(2)}}^2
    +
    \cdots
    +
    \sum_{i_N = R_N + 1}^{I_N} \parens*{\sigma_{i_N}^{(N)}}^2.
\end{align*}
\end{lemma}

\begin{proof}
Let the HOSVD of $\tensor{X} = \tensor{S} \times_1 \mat{U}^{(1)} \times_2 \cdots \times_N \mat{U}^{(N)}$ as in
\Cref{thm:hosvd}.
Let $\overline{\tensor{S}}$ denote the truncated version of $\tensor{S}$
with respect to core shape $\mat{r}$ such that
\[
    \overline{s}_{i_1 \dots i_N} =
    \begin{cases}
      s_{i_1 \dots i_N} & $\text{if $i_n \le R_n$ for all $n \in [N]$},$\\
      0 & \text{otherwise}.
    \end{cases}
\]
Let $\overline{\tensor{X}} = \overline{\tensor{S}} \times_1 \mat{U}^{(1)} \times_2 \dots \times_N \mat{U}^{(N)}$.
Summing over all dimensions $n \in [N]$ and using the HOSVD results in \cref{thm:hosvd},
\begin{align*}
    \norm{\tensor{S} - \overline{\tensor{S}}}_{\frobenius}^2
    &\leq 
    \sum_{n=1}^N
    \sum_{j = R_n + 1}^{I_n} \norm{\parens*{\tensor{S}-\overline{\tensor{S}}}_{i_n=j}}_{\frobenius}^2 \\
    &=
    \sum_{n=1}^N
    \sum_{j = R_n + 1}^{I_n} \norm{\tensor{S}_{i_n=j}}_{\frobenius}^2 \\
    &=
    \sum_{n=1}^N
    \sum_{i_n = R_n + 1}^{I_n} \parens*{\sigma_{i_n}^{(n)}}^2.
\end{align*}
We have
$\norm{\tensor{S} - \overline{\tensor{S}}}_{\frobenius}^2 = \norm{\mat{S}_{(n)} - \overline{\mat{S}}_{(n)}}_{\frobenius}^2$ for all values of $n$.
Further, since the Kronecker product of two orthogonal matrices is also orthogonal
and multiplication by an orthogonal matrix does not affect the Frobenius norm,
Equation~\Cref{eq:mode_n_mat} implies that
\begin{align*}
    \norm{\mat{S}_{(n)} - \overline{\mat{S}}_{(n)}}_{\frobenius}^2
    &=
    \norm{\mat{U}_{(n)} (\mat{S}_{(n)} - \overline{\mat{S}}_{(n)})(\mat{U}_{(1)} \ktimes \dots \ktimes \mat{U}_{(n-1)} \ktimes \mat{U}_{(n+1)} \ktimes \cdots \ktimes \mat{U}_{(N)})^\intercal}_{\frobenius}^2
\\ & =
\norm{\tensor{X} - \overline{\tensor{X}}}_{\frobenius}^2.
\end{align*}
Observing that $\overline{\tensor{X}} = \widehat{\tensor{X}}_{\textnormal{HOSVD}(\mat{r})}$
by the definition of \texttt{TuckerHOSVD} in \Cref{alg:tucker_hosvd}
and putting everything together,
\[
    L(\tensor{X}, \mat{r})
    \leq
    \norm{\tensor{X} - \widehat{\tensor{X}}_{\textnormal{HOSVD}(\mat{r})}}_{\frobenius}^2
    =
    \norm{\tensor{X} - \overline{\tensor{X}}}_{\frobenius}^2
    =
    \norm{\tensor{S} - \overline{\tensor{S}}}_{\frobenius}^2
    \leq
    \sum_{n=1}^N \sum_{i_n = R_n + 1}^{I_n} \parens*{\sigma_{i_n}^{(n)}}^2,
\]
which completes the proof.
\end{proof}

\subsubsection{Lower bounding the reconstruction error}

\begin{theorem}[\citet{eckart1936approximation}]
\label{thm:low_rank}
Let $\mat{A} \in \mathbb{R}^{n\times d}$ with $n \ge d$ and singular values
$\sigma_1 \geq \sigma_2 \geq \dots \geq \sigma_d \geq 0$.
Let $\mat{A}_k$ be the best rank-$k$ approximation of $\mat{A}$ in the Frobenius norm.
Then
\[
    \norm*{\mat{A} - \mat{A}_k}_{\frobenius}^2 = \sum_{i = k+1}^d \sigma_i^2.
\]
\end{theorem}

\begin{restatable}{lemma}{LowerBoundForF}
\label{lemma:lower_bound_for_f}
For any
$\tensor{X}\in\mathbb{R}^{I_1\times\cdots\times I_N}$
and $n \in [N]$,
let
$\sigma^{(n)}_{1} \ge \dots \ge \sigma^{(n)}_{I_n}$ denote the singular values of $\mat{X}_{(n)}$.
For every core shape $\mat{r} = (R_1, R_2, \dots, R_N) \in [I_1] \times [I_2] \times \dots \times [I_N]$,
we have
\begin{align*}
    \Loss(\tensor{X}, \mat{r})
    \geq
    \max_{n\in [N]} \sum_{i_n = R_n + 1}^{I_n} \parens*{\sigma_{i_n}^{(n)}}^2.
\end{align*}
\end{restatable}

\begin{proof}
Let the core tensor and factors that minimize $L(\tensor{X}, \mat{r})$ be
$\tensor{G} \in \R^{R_1 \times \dots \times R_N}$
and $\mat{A}^{(n)} \in \R^{I_n\times R_n}$,
i.e.,
\[
    \Loss(\tensor{X}, \mat{r})
    =
    \norm{\tensor{X} - \tensor{G}\times_1 \mat{A}^{(1)} \times_2 \dots \times_N \mat{A}^{(N)}}_{\frobenius}^2.
\]
Let $\widehat{\tensor{X}} = \tensor{G}\times_1 \mat{A}^{(1)} \times_2 \dots \times_N \mat{A}^{(N)}$.
Equation~\Cref{eq:mode_n_mat} and the dimensions of $\mat{A}^{(n)} \in \R^{I_n \times R_n}$
imply that
\[
    \rank\parens*{\widehat{\mat{X}}_{(n)}} \le \rank\parens*{\mat{A}^{(n)}} \le R_n,
\]
since $R_n \le I_n$.
The Eckart--Young--Mirsky theorem
(\Cref{thm:low_rank})
with the characterization of $\sigma_{i}^{(n)}$
in \Cref{thm:hosvd} gives
\begin{align}
\label{eqn:error_lower_bound}
    \norm{\tensor{X} - \widehat{\tensor{X}}}_{\frobenius}^2
=
\norm{\mat{X}_{(n)} - \widehat{\mat{X}}_{(n)}}_{\frobenius}^2 
\geq
    \sum_{i_n = R_n + 1}^{I_n} \parens*{\sigma_{i_n}^{(n)}}^2.
\end{align}
Equation~\Cref{eqn:error_lower_bound} holds for all values of $n$,
so take the equation that maximizes the right-hand side.
\end{proof}

\subsubsection{Combining the results}

Now we combine \Cref{lemma:upper_bound_for_loss}
and \Cref{lemma:lower_bound_for_f} to give an approximation inequality
that is true for all core shapes.

\begin{lemma}
\label{lemma:sandwich_loss}
For any $\tensor{X}\in\mathbb{R}^{I_1\times\cdots\times I_N}$
and any core shape $\mat{r} = (R_1, R_2, \dots, R_N) \in [I_1] \times [I_2] \times \dots \times [I_N]$, we have
\[
    \frac{1}{N} \cdot \SurrogateLoss(\tensor{X}, \mat{r})
    \le
    \Loss(\tensor{X}, \mat{r})
    \le
    \SurrogateLoss(\tensor{X}, \mat{r}),
\]
where $\SurrogateLoss(\tensor{X}, \mat{r})$ is defined in Equation~\Cref{eqn:surrogate_loss_def}.
\end{lemma}

\begin{proof}
Summing Equation~\Cref{eqn:error_lower_bound}
in the proof of~\Cref{lemma:lower_bound_for_f}
over all values of $n \in [N]$ gives
\begin{equation}
\label{eqn:sum_modewise_tail_errors}
    \sum_{n=1}^N
    \sum_{i_n = R_{n}+1}^{I_n} \parens*{\sigma_{i_n}^{(n)}}^2
    \le
    N \cdot \Loss(\tensor{X}, \mat{r})
    \implies
    \frac{1}{N} \cdot \SurrogateLoss(\tensor{X}, \mat{r})
    \le
    \Loss(\tensor{X}, \mat{r}).
\end{equation}
The upper bound
$\Loss(\tensor{X}, \mat{r}) \le \SurrogateLoss(\tensor{X}, \mat{r})$
is a restatement of \Cref{lemma:upper_bound_for_loss}.
\end{proof}

\HosvdTruncationError*

\begin{proof}
The proof follows by combining~\Cref{lemma:upper_bound_for_loss}
and
Equation~\Cref{eqn:sum_modewise_tail_errors}.
\end{proof}

\subsection{Proof of \Cref{lem:hosvd_packing_problem}}
\label{subsec:proxy_sandwhich}

\HosvdPackingProblem*

\begin{proof}
For any $n$, we have
\begin{align*}
    \sum_{i_n=1}^{I_n} \parens*{\sigma_{i_n}^{(n)}}^2
    =
    \norm{\mat{X}_{(n)}}_{\frobenius}^2
    =
    \norm{\tensor{X}}_{\frobenius}^2.
\end{align*}
Therefore, for any choice of $\mat{r} = (R_1, R_2, \dots, R_N)$, we have
\[
    \bracks*{\sum_{n=1}^N \sum_{i_n=1}^{R_n} \parens*{\sigma_{i_n}^{(n)}}^2}
    +
    \bracks*{\sum_{n=1}^N \sum_{i_n=R_n + 1}^{I_n} \parens*{\sigma_{i_n}^{(n)}}^2}
    =
    N \norm*{\tensor{X}}_{\frobenius}^2.
\]
This is a constant value that only depends on $\tensor{X}$,
so minimizing $\SurrogateLoss(\tensor{X}, \mat{r})$
is equivalent to maximizing the packing version
since both problems optimize over the same set $F$.

Lastly, we have
\begin{align*}
    \SurrogateLoss\parens*{\tensor{X}, \widetilde{\mat{r}}^*}
    \le
    \SurrogateLoss\parens*{\tensor{X}, \mat{r}^*}
    \le
    N \cdot \Loss\parens*{\tensor{X}, \mat{r}^*},
\end{align*}
where the first inequality follows from optimizing the surrogate loss
and the second inequality follows from \Cref{lemma:sandwich_loss} since
that result holds for all core shapes.
\end{proof}

\subsection{Hardness}
\label{app:hardness}

\begin{definition}
Let $N \ge 2 $ be an even integer and $w_1, \dots,w_N \ge 1$ be integers.
The EQUIPARTITION problem asks to determine whether there exists a subset $S\subseteq [N]$ of size $n/2$ such that
\[
\sum_{i \in S} w_i = \sum_{i \in [N] \setminus S} w_i.
\]
\end{definition}

\begin{lemma}[{\citealt[SP12]{garey1979computers}}]
\textnormal{EQUIPARTITION} is \textnormal{NP-complete}.
\end{lemma}

We now give a reduction from the equipartition problem to the Tucker packing problem.

\Hardness*

\begin{proof}
Let $T, w_1,\dots, w_T$ be an instance of EQUIPARTITION
where $w_n \geq 2$ for all $n \in [T]$.
Notice that the assumption $w_n \geq 2$ is without loss of generality because
we can multiply all of the values $w_1,\dots, w_T$ by two.

Let $M = \sum_{n \in [T]} w_n$ be the sum of all weights,
and let $N \geq T$ be the smallest integer such that $2^{N-T/2} > 4(N-T) + 3M/2$.
Now we construct an instance of the Tucker packing problem.
For each $n\in[T]$, let:
\begin{itemize}
    \item $I_n = w_n$
    \item $a^{(n)}_1 = 2M$
    \item $a^{(n)}_2 = M+w_n$
    \item $a^{(n)}_{i_n}=0$, for all $i_n \in [I_n]\setminus\{1,2\}$
\end{itemize}
Next, for each $n \in [N]\setminus[T]$, let:
\begin{itemize}
    \item $I_n = 2$
    \item $a^{(n)}_1 = a^{(n)}_2 = 2M$
\end{itemize}
Finally, set the budget to be $c = 2^{N-T/2} + 4(N-T) + 3M/2$.

First, notice that this is a valid instance of the Tucker packing problem
since $a^{(n)}_1 \geq a^{(n)}_2 \geq \cdots \geq a^{(n)}_{I_n}$ for all $n\in[N]$.
Further, since $N = O(T \cdot \log_2(3M/2))$ and $a^{(n)}_{i_n} = 0$ for $i_n \in [I_n]\setminus\{1,2\}$,
the size of the description of this problem is polynomial in the size of the description of the corresponding EQUIPARTITION problem.

Now we consider a decision version of this Tucker packing problem in which we are asked to determine whether there exists a feasible solution $(R_1,\ldots,R_N)$ such that
\begin{equation}
\label{eqn:reduction_decision_question}
    \sum_{n=1}^N \sum_{i_n=1}^{R_n} a_{i_n}^{(n)} \geq M(4N-3T/2)+ M/2.
\end{equation}

We show that a positive answer to the decision version of the Tucker packing problem in ~\Cref{eqn:reduction_decision_question}
implies a positive answer to the EQUIPARTITION problem and vice versa.

Suppose the answer to the decision version of the Tucker packing problem is YES, and $\mat{r}^*$ is an optimal solution such that 
\[
    \sum_{n=1}^N \sum_{i_n=1}^{R^*_n} a_{i_n}^{(n)} \geq M(4N-3T/2)+ M/2
    \quad ~~\text{and}~~ \quad
    \prod_{n=1}^N R^*_n + \sum_{n=1}^N I_n R^*_n \le c.
\]
Since
\[
    c = 2^{N-T/2} + 4(N-T) + 3M/2 < 2\cdot 2^{N-T/2} = 2^{N-T/2 + 1},
\]
there are at most $N-T/2$ values of $R^*_n$ such that $R^*_n \ge 2$.
Further, since $a^{(n)}_{i_n}=0$ for all $i_n \ge 3$,
we never have $R^*_n > 2$ in a minimal optimal solution.
It follows that $R^*_n \in \{1, 2\}$ for all $n \in [N]$, and
\[
    \prod_{n=1}^N R^*_n \leq 2^{N-T/2}.
\]

Next, we establish the structure of an optimal solution to this Tucker packing instance.
Observe that $\widehat{\mat{r}}=(\widehat{R}_1,\dots,\widehat{R}_N)$ with 
$\widehat{R}_1 = \dots = \widehat{R}_T = 1$
and
$\widehat{R}_{T+1} = \dots = \widehat{R}_N=2$
is a feasible solution $\mat{r}$ that achieves an objective value of $M(4N-2T)$. 
Now consider any feasible solution in which there exists
$i\in[T]$ and $j\in[N]\setminus[T]$
such that $R_i = 2$ and $R_j=1$.
If we switch the values of $R_i$ and $R_j$,
then the cost decreases by $w_i - 2 \geq 0$ (i.e., the solution is still feasible),
and the objective value increases by $M-w_{i}>0$.
Therefore, since $\widehat{\mat{r}}$ is feasible,
in an optimal solution
we have $R_n = 2$ for all $n \in [N]\setminus [T]$
and at most $T/2$ of the $R_n$'s for $n \in [T]$ are equal to two.

Let $S=\{i\in [T]: R^*_i=2\}$.
Then by construction we have
\[
    \sum_{n\in S} a^{(n)}_2
    =
    \sum_{n \in S} (M + w_{n})
    =
    M \abs{S} + \sum_{n\in S} w_n < M(\abs{S} + 1).
\]

Moreover, since the answer to the decision problem is YES
and in an optimal solution we have $R_n^* = 2$ for all $n\in [N]\setminus[T]$,
it follows that
\begin{align}
\label{eq:a2-bound}
    \sum_{n\in S} a^{(n)}_2
    &=
    f(\mat{r}^*) - \sum_{n = 1}^N a_{1}^{(n)} - \sum_{n = T+1}^N a_{2}^{(n)} \\
    &=
    f(\mat{r}^*) - 2NM - 2(N-T)M \notag \\
    &\ge
    M(4N - 3T/2) + M/2 - 2NM - 2(N-T)M \notag \\
    &=
    MT/2 + M/2. \notag
\end{align}
Therefore,
\[
    M(\abs{S}+1) > MT/2+M/2,
\]
which implies
$\abs{S} > T/2-1/2$,
so $\abs{S} \geq T/2$ since $\abs{S}$ and $T/2$ are integers.
Using the characterization above about an optimal solution
together with the fact that the budget is strictly less than $2^{N-T/2 + 1}$
gives us $\abs{S} \leq T/2$.
Thus, a YES to the decision problem
implies that $\abs{S} = T/2$,
which further
implies $\prod_{n=1}^N R^*_n = 2^{N-T/2}$.

It then follows from our choice of budget $c$ that
\[
    \sum_{n=1}^N I_n R^*_n \leq c - 2^{N-T/2} = 4(N-T) + 3M/2,
\]
which then by the definition of $I_n$ implies that
\[
    \parens*{\sum_{n=1}^T w_n R_{n}^* + \sum_{n=T+1}^N 2 R_{n}^*}
    =
    \parens*{M + \sum_{n \in S}w_n} + 4(N - T)
    \le 4(N - T) + 3M/2
    \implies
    \sum_{n\in S} w_n \leq M/2.
\]
Furthermore,
using \eqref{eq:a2-bound}, the definition of the $a_{i_n}^{(n)}$'s,
and the fact that $\abs{S}=T/2$, we have
\[
    \sum_{n \in S} a_{2}^{(n)}
    =
    \sum_{n \in S} (M + w_n)
    =
    |S| M + \sum_{n \in S} w_n
    \ge MT/2 + M/2
    \implies
    \sum_{n\in S} w_n \geq M/2.
\]
Putting everything together,
we get $\sum_{n\in S} w_n = M/2$.
Thus, $S$ is a solution for the EQUIPARTITION problem.

Now suppose the answer to the EQUIPARTITION problem is YES.
Let $S \subseteq [T]$ such that $|S| = T/2$ and $\sum_{n\in S} w_n = M/2$.
Construct $\mat{r}^*$ as follows:
For each $n \in S \cup ([N] \setminus [T])$, set $R^*_n = 2$;
for each $n \in [T] \setminus S$, set $R^*_n = 1$.

Then, by the definitions of $I_n$ and $a_{i_n}^{(n)}$ above, we have
\[
    \sum_{n=1}^N \sum_{i_n=1}^{R^*_n} a_{i_n}^{(n)} \geq M(4N-3T/2)+ M/2
    \quad ~~\text{and}~~ \quad
    \prod_{n=1}^N R^*_n + \sum_{n=1}^N I_n R^*_n \le c,
\]
which completes the proof. \qedhere

\end{proof}

\subsection{Translating between approximate maximization and minimization}
\label{app:complement_approximation_error}

We prove an additive-error guarantee that shows how a $(1-\varepsilon')$-approximate
solution to the Tucker packing problem, i.e., a core shape $\mat{r} \in [I_1] \times \dots \times [I_N]$,
can lead to an increase in the surrogate loss objective.

\ComplementApproximationError*

\begin{proof}
Let $\varepsilon' = \varepsilon / N$
and $\widetilde{\mat{r}}^*$ be the optimal shape for the surrogate loss $\SurrogateLoss$.
If $\mat{r} = (R_1, \dots, R_N)$ is a $(1 - \varepsilon')$-approximation to the Tucker packing problem,
it follows that
\begin{align*}
    \frac{\widetilde{L}(\tensor{X}, \mat{r})}{\norm*{\tensor{X}}_{\frobenius}^2}
    &=
    \frac{N \norm*{\tensor{X}}_{\frobenius}^2
    - \sum_{n=1}^N \sum_{i_n = 1}^{R_n} \parens*{\sigma_{i_n}^{(n)}}^2}{\norm*{\tensor{X}}_{\frobenius}^2} \\
    &\le
    \frac{N \norm*{\tensor{X}}_{\frobenius}^2
    - \parens*{1 - \varepsilon'} \sum_{n=1}^N \sum_{i_n = 1}^{\widetilde{R}^*_n} \parens*{\sigma_{i_n}^{(n)}}^2}{\norm*{\tensor{X}}_{\frobenius}^2} \\
    &=
    \frac{\SurrogateLoss\parens*{\tensor{X}, \widetilde{\mat{r}}^*}}{\norm*{\tensor{X}}_{\frobenius}^2}
    +
    \varepsilon' \parens*{N - \frac{\SurrogateLoss\parens*{\tensor{X}, \widetilde{\mat{r}}^*}}{\norm*{\tensor{X}}_{\frobenius}^2}} \\
    &\le
    \frac{\SurrogateLoss\parens*{\tensor{X}, \widetilde{\mat{r}}^*}}{\norm*{\tensor{X}}_{\frobenius}^2}
    +
    \varepsilon.
\end{align*}

\Cref{thm:tensor_subspace_approximation} gives us $\Loss(\tensor{X}, \mat{r}) \leq \SurrogateLoss(\tensor{X}, \mat{r}) \leq N \cdot \Loss(\tensor{X}, \mat{r})$.
By definition $\SurrogateLoss\parens*{\tensor{X}, \widetilde{\mat{r}}^*} \leq \SurrogateLoss\parens*{\tensor{X}, \mat{r}^*}$, so we have
\[
    \textnormal{RRE}(\tensor{X}, \mat{r}) = 
    \frac{\Loss(\tensor{X}, \mat{r})}{\norm*{\tensor{X}}_{\frobenius}^2}
    \le \frac{\widetilde{L}(\tensor{X}, \mat{r})}{\norm*{\tensor{X}}_{\frobenius}^2} \leq 
    \frac{\SurrogateLoss\parens*{\tensor{X}, \widetilde{\mat{r}}^*}}{\norm*{\tensor{X}}_{\frobenius}^2}
    +
    \varepsilon \leq 
    \frac{\SurrogateLoss\parens*{\tensor{X}, \mat{r}^*}}{\norm*{\tensor{X}}_{\frobenius}^2}
    +
    \varepsilon\leq 
    N \cdot \frac{\Loss\parens*{\tensor{X}, \mat{r}^*}}{\norm*{\tensor{X}}_{\frobenius}^2}
    +
    \varepsilon =
    N \cdot \textnormal{RRE}(\tensor{X}, \mat{r}^*) + \varepsilon,
\]
as desired.
\end{proof}
\section{Missing analysis from \Cref{sec:algorithm}}
\label{app:algorithm}

\subsection{Proof of \Cref{lem:grid-search}}

\GridSearch*

\begin{proof}
Let $k_n \ge 0$ be the largest integer such that $(1+\varepsilon)^{k_n} \leq R^*_n$
for each $n\in[N]$.
Further, let
\[
    \widehat{R}_{n} = \ceil*{(1+\epsilon)^{k_{n}}}.
\]
Since $R^*_n$ is an integer, we know that $\widehat{R}_{n} \leq R^*_n$.
Therefore, because $F$ is downwards closed,
$\widehat{\mat{r}} = (\widehat{R}_{1},\dots,\widehat{R}_{N})$
is a feasible solution to \eqref{eq:comb_opt_whole_rep_grid_prob}.
It follows that
\begin{align}
\label{eqn:approximate_objective_00}
    f(\widehat{\mat{r}}) \leq f(\mat{r}^{(\varepsilon)}).
\end{align}

Now we will show that $f(\mat{r}^*) \leq (1+\varepsilon)f(\widehat{\mat{r}})$.
Since $a_{1}^{(n)} \geq \dots \geq a_{I_n}^{(n)} \geq 0$ for all $n \in [N]$, we have
\begin{equation}
\label{eqn:approximate_objective_01}
    (1+\epsilon) \sum_{i_n = 1}^{\widehat{R}_n} a_{i_n}^{(n)}
    \geq
    \sum_{i_n = 1}^{\floor{(1+\epsilon) \widehat{R}_n }} a_{i_n}^{(n)}.
\end{equation}
By the definition of $k_n$, it follows that
\begin{align*}
    (1+\epsilon) \widehat{R}_n
    &=
    (1+\epsilon) \ceil*{(1+\epsilon)^{k_{n}}} \\
    &\geq
    (1+\epsilon)^{k_n+1} \\
    &>
    \ropti{n}.
\end{align*}
Since $\ropti{n}$ is an integer, we have
$\floor{(1+\epsilon) \widehat{R}_n} \geq \ropti{n}$.
Therefore, using Equation~\Cref{eqn:approximate_objective_01} we have
\begin{equation}
\label{eqn:approximate_objective_01}
    (1+\epsilon)\sum_{i_n = 1}^{\widehat{R}_n} a_{i_n}^{(n)}
    \geq
    \sum_{i_n = 1}^{\ropti{n}} a_{i_n}^{(n)}.
\end{equation}
Finally, summing over $n\in[N]$ and using Equation~\Cref{eqn:approximate_objective_00} gives us
\[
    (1+\epsilon)f(\mat{\widehat{r}}) \geq f(\mat{r}^*)
    \implies
    f(\mat{r}^{(\varepsilon)}) \ge (1+\varepsilon)^{-1} f(\mat{r}^*).
\]

\paragraph{Algorithm.}
Now we design and analyze a simple algorithm to solve the grid-search problem in~\Cref{eq:comb_opt_whole_rep_grid_prob}.
First observe that
\[
    |S_n^{(\varepsilon)}|
    =
    1 + \floor*{\log_{1+\varepsilon}(I_n)}
    =
    1 + \floor*{\frac{\log_2(I_n)}{\log_2(1+\varepsilon)}}.
\]
It follows that the number of feasible solutions for Problem~\eqref{eq:comb_opt_whole_rep_grid_prob} is
\[
    O\parens*{\prod_{n=1}^N \parens*{1 + \frac{\log_2(I_n)}{\log_2(1 + \varepsilon)}} }
    =
    O\parens*{\frac{1}{\log_2^N(1 + \varepsilon)} \prod_{n=1}^N \parens*{1 + \log_2(I_n)} },
\]
since $\log_2(1+\varepsilon) \le 1$ for $\varepsilon \le 1$.
Further, observing that $\log_2(1+\varepsilon) \geq \varepsilon$ for $0 < \varepsilon \leq 1$ implies a bound of
\[
    O\parens*{\frac{1}{\varepsilon^N} \prod_{n=1}^n \parens*{1 + \log_2(I_n)}}
\]
on the number of feasible solutions.

After computing the prefix sums for elements of $S_n^{(\varepsilon)}$'s,
we can iterate over the elements of $S_1^{(\varepsilon)} \times \cdots \times S_N^{(\varepsilon)}$
in the lexicographical order,
check for feasibility, and compute the objective value $f(\mat{r})$
for any candidate solution in amortized time $O(1)$.
Finally, note that the prefix sums for elements of $S_n^{(\varepsilon)}$ can be computed in $O(I_n)$ time.
\end{proof}

\paragraph{Rationale for two different types of budget splits in \Cref{alg:total-size-Tucker}.}
Instead of considering two different budget splitting methods
(i.e., with $c_{\text{factor}}$ and $c_{\text{core}}$),
one could just consider the budget splits
$c_{\text{factor}} \in [\sum_{n=1}^N I_n^2]$
over problems of the form:
\begin{align*}
            \text{maximize~~} & \MaxObj(\mat{r}) \\ 
            \nonumber \text{subject to~~}
                & \textstyle \prod_{n=1}^N R_n \le c - c_{\text{factor}} \\
                \nonumber
                & \textstyle \sum_{n=1}^N I_n R_n \le c_{\text{factor}} \\
                \nonumber
                & R_n \in [I_n] \hspace{1cm} \forall n\in [N]
\end{align*}
This approach also gives a PTAS.
The running time, however, is $O\parens{\parens{\sum_{n=1}^N I_n^2} \parens{\sum_{n=1}^N I_n}^{O\parens*{1/\varepsilon}}}$. In contrast, the running time of \Cref{alg:total-size-Tucker} is $O\parens{\parens{\log_{1+\epsilon}(c) + \ceil*{1/\varepsilon}\sum_{n=1}^N I_n} \parens{\sum_{n=1}^N I_n}^{O\parens*{1/\varepsilon}}}$.
Note that $c \le \sum_{n=1}^N I_n^2 + \prod_{n=1}^N I_n$.
Therefore, $\log_{1+\epsilon}(c) = O(\frac{1}{\epsilon}\sum_{n=1}^N \log_2(I_n)) \ll \ceil*{1/\varepsilon}\sum_{n=1}^N I_n$,
where the first equality follows from the concavity of the logarithm function.
Thus, our running time is $O\parens{\parens{ \ceil*{1/\varepsilon}\sum_{n=1}^N I_n} \parens{\sum_{n=1}^N I_n}^{O\parens*{1/\varepsilon}}}$,
which is always smaller than $O\parens{\parens{\sum_{n=1}^N I_n^2} \parens{\sum_{n=1}^N I_n}^{O\parens*{1/\varepsilon}}}$.

\section{Tree tensor networks}
\label{app:tree-tensor-networks}

\begin{remark}
Although the tree tensor network is considered for general Hilbert spaces and has been defined in full generality using the tensor network notations (see, e.g., \citet[Chapter~11]{hackbusch2019tensor} and \citet[Chapter~3]{kramer2020tree}),
we consider finite-dimensional Euclidean spaces to keep our notation simple.
\end{remark}

\begin{remark}
\cref{thm:tree-tensor-packing} is implied by \cref{lem:grid-search}.
\end{remark}

\begin{figure}[H]
    \centering
    \begin{tikzpicture}[
roundnode/.style={circle, draw=green!60, fill=green!5, very thick, minimum size=5mm},
squarednode/.style={rectangle, draw=red!60, fill=red!5, very thick, minimum size=5mm},
]
\node[squarednode]      (in) {$v$};
\node[roundnode]        (leaf2) [below left=1 cm and 0.3 cm of in] {2};
\node[roundnode]        (leaf3) [below right=1 cm and 0.3 cm of in] {3};
\node[roundnode]        (leaf1) [left=of leaf2] {1};
\node[roundnode]        (leaf4) [ right=of leaf3] {4};

\draw[-] (leaf1.north)  -- node[midway,above]{$e_1$} (in.south);
\draw[-] (leaf2.north)  -- node[midway,left]{$e_2$} (in.south);
\draw[-] (leaf3.north)  -- node[midway,right]{$e_3$} (in.south);
\draw[-] (leaf4.north)  -- node[midway,above]{$e_4$} (in.south);
\end{tikzpicture}
    \caption{Tree tensor network corresponding to a Tucker decomposition.}
    \label{fig:tree-tensor-Tucker}
    \vspace{-0.20cm}
\end{figure}
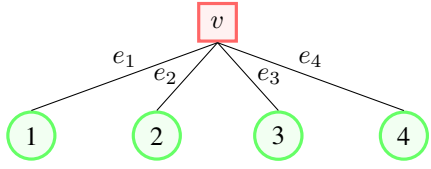

\paragraph{Example (Tucker decomposition).}
Tucker decomposition corresponds to a tree tensor of depth one.
For example, the tree tensor in \cref{fig:tree-tensor-Tucker} consists of matrices $\mat{A}_1\in\R^{I_1\times R_{e_1}},\mat{A}_2\in\R^{I_2\times R_{e_2}},\mat{A}_3\in\R^{I_3\times R_{e_3}},\mat{A}_4\in\R^{I_4\times R_{e_4}}$, and tensor $\tensor{T}_{v}\in\R^{R_{e_1}\times R_{e_2}\times R_{e_3}\times R_{e_4}}$. The corresponding reconstruction is \[
\widehat{\tensor{X}} =  \tensor{T}_v\times_{e_4} \mat{A}_4\times_{e_3} \mat{A}_3\times_{e_2} \mat{A}_2\times_{e_1} \mat{A}_1.
\]

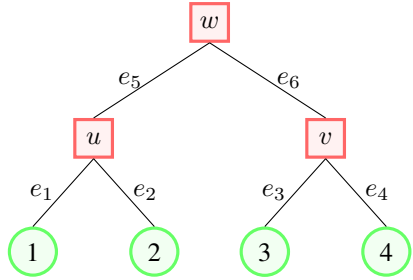
\begin{figure}[H]
    \vspace{-0.20cm}
    \centering
    \begin{tikzpicture}[
roundnode/.style={circle, draw=green!60, fill=green!5, very thick, minimum size=5mm},
squarednode/.style={rectangle, draw=red!60, fill=red!5, very thick, minimum size=5mm},
]
\node[squarednode]      (in3) {$w$};
\node[squarednode]      (in1) [below left=1 cm and 1 cm of in3] {$u$};
\node[squarednode]      (in2) [below right=1 cm and 1 cm of in3] {$v$};
\node[roundnode]        (leaf1) [below left=1 cm and 0.3 cm of in1] {1};
\node[roundnode]        (leaf2) [below right=1 cm and 0.3 cm of in1] {2};
\node[roundnode]        (leaf3) [below left=1 cm and 0.3 cm of in2] {3};
\node[roundnode]        (leaf4) [below right=1 cm and 0.3 cm of in2] {4};

\draw[-] (leaf1.north)  -- node[midway,left]{$e_1$} (in1.south);
\draw[-] (leaf2.north)  -- node[midway,right]{$e_2$} (in1.south);
\draw[-] (leaf3.north)  -- node[midway,left]{$e_3$} (in2.south);
\draw[-] (leaf4.north)  -- node[midway,right]{$e_4$} (in2.south);
\draw[-] (in1.north)  -- node[midway,left]{$e_5$} (in3.south);
\draw[-] (in2.north)  -- node[midway,right]{$e_6$} (in3.south);
\end{tikzpicture}
    \caption{Tree tensor network example corresponding to a hierarchical Tucker decomposition.}
    \label{fig:tree-tensor}
    \vspace{-0.20cm}
\end{figure}

\paragraph{Example (Hierarchical Tucker decomposition).}
To to better understand~\cref{def:tree-tensor}, we give an example for a tensor of order $4$. Consider the tree illustrated in \cref{fig:tree-tensor}. This tree tensor network corresponds to matrices $\mat{A}_1\in\R^{I_1\times R_{e_1}},\mat{A}_2\in\R^{I_2\times R_{e_2}},\mat{A}_3\in\R^{I_3\times R_{e_3}},\mat{A}_4\in\R^{I_4\times R_{e_4}}$, and tensors $\tensor{T}_{u}\in\R^{R_{e_1}\times R_{e_2}\times R_{e_5}},\tensor{T}_{v}\in\R^{R_{e_3}\times R_{e_4}\times R_{e_6}},\tensor{T}_{w}\in\R^{R_{e_5}\times R_{e_6}}$.
The corresponding reconstruction is
\[
    \widehat{\tensor{X}} = \tensor{T}_w \times_{e_6} \tensor{T}_v \times_{e_5} \tensor{T}_u\times_{e_4} \mat{A}_4\times_{e_3} \mat{A}_3\times_{e_2} \mat{A}_2\times_{e_1} \mat{A}_1.
\]

\vspace{-0.75cm}
\section{Experiments}
\label{app:experiments}

We provide short descriptions about each of the tensor datasets used in 
the core shape experiments in \Cref{sec:experiments},
which gives some insight into why the algorithms build the core shapes they do.

\paragraph{Cardiac MRI.} $256 \times 256 \times 14 \times 20$
tensor whose elements are MRI measurements indexed by $(x,y,z,t)$,
where $(x,y,z)$ is a point in space and $t$ corresponds to time.

\paragraph{Hyperspectral.} $1024 \times 1344 \times 33$ tensor of
time-lapse hyperspectral radiance images of a nature scene
that is undergoing illumination changes~\citep{nascimento2016spatial}.

\paragraph{VicRoads.} $1084 \times 2033 \times 96$ tensor containing
2033 days of traffic volume data from Melbourne and its surrounding suburbs.
This data comes from a network of 1084 road sensors
measured in 15 minute intervals~\citep{schimbinschi2015traffic}.

\paragraph{COIL-100.} $7200 \times 128 \times 128 \times 3$
tensor containing
7200 colored photos of 100 different objects (72 images per object)
taken at 5-degree rotations~\citep{nene1996columbia}.
This is a widely-used dataset in the computer vision research community.

\end{document}